\documentclass[11pt,a4paper,reqno]{amsart}
\usepackage[hmargin={2.7cm,2.7cm},vmargin={2.5cm,2.5cm},centering]{geometry}
\usepackage[lite]{amsrefs}
\usepackage{amssymb,amsmath,amsthm,amsfonts,amscd}
\usepackage{graphicx}
\usepackage[dvipsnames]{xcolor}
\usepackage{mathrsfs}
\usepackage[unicode,psdextra]{hyperref}
\usepackage[utf8]{inputenc}
\hypersetup{pdfborder={0 0 0.06},citebordercolor=[rgb]{0.8196,0.2275,0.5098},linkbordercolor=[rgb]{0.1765,0.5490,0.8118},urlbordercolor=[rgb]{0.7059,0.5333,0.1137}}
\usepackage{tikz-cd}
\usepackage[english]{babel}
\usepackage{dsfont}
\usepackage{url}
\usepackage{enumerate}
\usepackage{setspace}
\usepackage[font={scriptsize,bf}]{caption}
\usepackage{changepage}
\usepackage{booktabs}
\usepackage[sort&compress,capitalise,nameinlink]{cleveref}
\crefname{section}{\textsection}{\textsection}
\crefname{subsection}{\textsection}{\textsection}
\crefname{subsubsection}{\textsection}{\textsection}
\crefname{paragraph}{\textparagraph}{\textparagraph}
\crefname{thm}{Theorem}{Theorems}
\crefname{proposition}{Proposition}{Propositions}

\DeclareMathOperator{\Tr}{Tr}

\renewcommand{\Im}{\mathrm{Im}}
\renewcommand{\Re}{\mathrm{Re}}
          \newtheorem{thm}{Theorem}[section]
          \newtheorem{proposition}[thm]{Proposition}
          \newtheorem{lemma}[thm]{Lemma}
          \newtheorem{corollary}[thm]{Corollary}
          \newtheorem{definition}[thm]{Definition}
          
          \theoremstyle{definition}

          \newtheorem{remark}[thm]{Remark}
          
          \newtheorem*{assumption*}{Assumption}

\renewcommand{\leq}{\leqslant}
\renewcommand{\geq}{\geqslant}

\usepackage{csquotes}
\allowdisplaybreaks

\begin{document}
\tolerance=2000 \setlength{\emergencystretch}{1em}

\title{Semiclassical limit of entropies and free energies}

\author[Z.\ Ammari]{Zied Ammari} \address{Université Marie \& Louis
  Pasteur\\Laboratoire de Mathématiques LmB\\UFR Sciences et techniques\\16
  route de Gray\\25030 Besançon CEDEX\\France}
\email{zied.ammari@univ-fcomte.fr}

\author[M.\ Correggi]{Michele Correggi} \address{Dipartimento di
  Matematica\\Politecnico di Milano\\P.zza Leonardo da Vinci 32\\20133,
  Milano\\Italy} \email{michele.correggi@gmail.com}
\urladdr{https://sites.google.com/view/michele-correggi}

\author[M.\ Falconi]{Marco Falconi} \address{Dipartimento di
  Matematica\\Politecnico di Milano\\P.zza Leonardo da Vinci 32\\20133,
  Milano\\Italy} \email{marco.falconi@polimi.it}
\urladdr{https://www.mfmat.org/}

\author[R.\ Gautier]{Raphaël Gautier} \address{Dipartimento di
  Matematica\\Politecnico di Milano\\P.zza Leonardo da Vinci 32\\20133,
  Milano\\Italy\\\& Université de Rennes\\ Campus de Beaulieu\\263, avenue du
  Général Leclerc\\35042 RENNES CEDEX\\France}
\email{raphael.gautier@polimi.it}
\email{raphael.gautier@univ-rennes.fr}

\date{\today}



\begin{abstract}
  Entropy and free energy are central concepts in both statistical physics
  and information theory, with quantum and classical facets. In mathematics
  these concepts appear quite often in different contexts (dynamical systems,
  probability theory, von Neumann algebras, \emph{etc.}). In this work, we
  study the von Neumann and Wehrl entropies from the point of view of
  semiclassical analysis. We first prove the semiclassical convergence of the
  von Neumann to the Wehrl entropy for quantum Gibbs states (thermal
  equilibrium), after a suitable renormalization has been taken into
  account. Then, we show that, in the same limit, the free energy functional
  defined with the Wehrl entropy $ \Gamma-$converges to its classical counterpart,
  so implying convergence of the minima and the associated minimizers.
\end{abstract}

\maketitle

\onehalfspacing{}

\section{Introduction and Main Results}
\label{sec:intr-main-results}

In recent years, there has been increasing interest in analyzing the
transition and the correspondence from quantum to classical in many-body
theory and in quantum information, in the framework of statistical
mechanics. For instance, the correlation functions of quantum Gibbs states of
Bose gases at positive temperature are related in a high temperature limit to
moments of classical Gibbs measures of a nonlinear Schr\"odinger equation
\cite{FKSS,Lewin2015}. At zero temperature, semiclassical higher order
expansions are proved for the ground state energy of Bose gases in the mean
field limit (see, e.g., \cite{BPS2021} and references therein). On the other
hand, the decoherence mechanism, one of the main issue in the realization of
quantum devices, interpreted as the relaxation of a quantum system to a
classical one, occurs as consequence of interactions with the environment
(for example, a thermal or a coherent-photons reservoir). However, the amount
of coherence loss is deeply related to quantum nature of the external
reservoir and varies in a highly non-trivial way in a semiclassical regime
\cite{CFFM}.

In the scaling limits and correspondences mentioned above, ground, excited
and equilibrium states play a central role, but equally important is the
notion of entropy that tracks the behavior of quantum and classical systems
and describe how far they are from thermodynamic equilibrium.  Our main goal
is precisely to study the transition between quantum entropies and classical
ones in the semiclassical regime.  As a simple yet relevant illustration, we
consider a class of toy models of quantum many-body systems or quantum fields
with finitely many degrees of freedom (\emph{i.e.}, with finite-dimensional
phase-space), and focus on the limit of the von Neumann and Wehrl quantum
entropies when the natural semiclassical parameter $\varepsilon\equiv \hbar $ or $\varepsilon\equiv\frac{1}{N}$
goes to zero. Our approach is motivated by several concrete applications to
the Bose-Hubbard model \cite{2024arXiv240504055A,MR2012977,MR4415671}, to
spin systems \cite{MR4153847,MR349181,MR4756728} and to Kubo-Martin-Schwinger
(KMS)states in statistical mechanics
\cite{AGGL,MR4289905,MR4756728}.

Generally speaking, it is well-known that the von Neumann entropy behaves
badly in the classical limit.  This motivated Wehrl to introduce a more
suitable semiclassical entropy in \cite{MR0496300,MR574951}, which then led
to much research on Wehrl's conjecture (see, e.g., \cite{FNT,MR3554894,NRT}
and references therein).  In this note, we first show that for Gibbs states
of generic interacting Hamiltonians both the von Neumann and Wehrl entropies
diverge in the semiclassical limit, but, up to suitable renormalization, both
converge to the Boltzmann entropy of the related classical system. Secondly,
we consider the free energy functional defined in terms of the Wehrl entropy
and address its behavior in the semiclassical limit. We establish its
$\Gamma$-convergence to its classical analogue, and, as a consequence, we
obtain the convergence of minima and minimizers.  From a wider perspective,
our results highlight the fact that quantum and classical entropies are
connected as far as thermal equilibrium states are concerned, so that, the
amount of randomness or chaos in quantum or classical equilibria is
comparable, and stays so, along the convergence. While the out of equilibrium
situation looks totally different, it would be interesting to understand more
the quantum-classical correspondence and transition near thermal equilibrium.

\subsection{Gibbs states}
\label{sec:gibbs-states}

Let us start by recalling that, given a system with classical
Hamiltonian $h: \mathcal{H} \to \mathbb{R} $, we define the \emph{classical
  Gibbs state} to be the probability measure on the (finite dimensional
linear) phase space $\mathcal{H} \cong \mathbb{C}^{d} $ with density $
(Z_{\beta,0})^{-1} e^{-\beta h}$, where
\begin{equation}
  \label{eq: zb}
  Z_{\beta, 0} : = \int_{\mathcal{H}} e^{-\beta h(z)} \mathrm{d} z,
\end{equation}
provided the latter quantity is finite; the classical Gibbs state describes
thermal equilibrium of the system at inverse temperature $\beta >0$. Its
quantum counterpart is the {\it quantum Gibbs state} formally defined as
$(Z_{\beta, \varepsilon})^{-1} e^{-\beta H_\varepsilon}$, where
\begin{equation}
  \label{eq: zbe}
  Z_{\beta, \varepsilon} : = \Tr_{\mathcal{H}_{\varepsilon}} e^{-\beta H_\varepsilon}
\end{equation}
needs to be finite, $H_\varepsilon$ is a suitable quantization of the energy
$h$ and $ \mathcal{H}_{\varepsilon} $ the Hilbert space of the quantum
system. Here and in the following $\varepsilon$ denotes the semiclassical
parameter.

\subsection{Classical and Quantum Entropies}
\label{sec:quant-class-entr}

Let us now focus on the various concepts of entropy. Let $\mu$ be a probability
measure over $ \mathcal{H} $, which is absolutely continuous with respect to
the Lebesgue measure, its {\it Boltzmann entropy} is defined as
\begin{equation}
  \label{eq:B}
  S_{\mathrm{B}}(\mu) := -\int_{\mathcal{H}} \frac{\mathrm{d} \mu}{\mathrm{d} z} \log \left(  \frac{\mathrm{d} \mu}{\mathrm{d} z}  \right) \mathrm{d} z.
\end{equation}
This takes also the name of continuous or differential entropy in information
theory, and although it is not the optimal definition from an information
theoretical viewpoint (most notably, it is not always positive), it suffices
for our purposes. The closely related notion of \emph{relative entropy} (the
Kullback-Leibler divergence in information theory) is defined for any couple
of probabilities $\mu,\nu$ on $\mathcal{H}$ as:
\begin{equation}
  \label{eq:rB}
  S_{\mathrm{B}}(\mu\Vert\nu)=
  \begin{cases}
    \, \displaystyle\int_{\mathcal{H}}^{}\log\left(\frac{\mathrm{d}\mu}{\mathrm{d}\nu}\right)   \mathrm{d}\mu & \text{if }\mu\ll\nu\\
    \,+\infty &\text{otherwise}
  \end{cases}\;.
\end{equation}
The relative entropy is \emph{always non-negative}, and it equals zero if and
only if $\mu=\nu$.

The natural quantum counterpart of Boltzmann entropy is the {\it von Neumann
  entropy}: naturally, the integral is replaced by the trace over the quantum
Hilbert space $ \mathcal{H}_{\varepsilon} $, the expression of the entropy
for a quantum state $\rho_\varepsilon \in \mathfrak{S}_+^1$ (positive,
trace-class operator with unit trace) thus becoming
\begin{equation}
  \label{eq:vN}
  S_{\mathrm{vN},\varepsilon}(\rho_\varepsilon) := - \Tr_{\mathcal{H}_{\varepsilon}} (\rho_\varepsilon \log(\rho_\varepsilon))\;.
\end{equation}
The \emph{von Neumann relative entropy} is also the natural generalization of
the classical relative entropy: for any
$\rho_{\varepsilon},\sigma_{\varepsilon}\in \mathfrak{S}_+^1$,
\begin{equation}
  \label{eq:rvN}
  S_{\mathrm{vN},\varepsilon}(\rho_{\varepsilon}\Vert\sigma_{\varepsilon}):= \Tr_{\mathcal{H}_{\varepsilon}}\Bigl(\rho_{\varepsilon}\bigl(\log(\rho_{\varepsilon})- \log(\sigma_{\varepsilon})\bigr)\Bigr)\;.
\end{equation}
Again, the von Neumann relative entropy is non-negative, and equals to zero
if and only if $\rho_{\varepsilon}=\sigma_{\varepsilon}$.

Pursuing the idea of defining a ``most classical'' quantum entropy, Wehrl
proposed a definition of entropy that is intermediate between the quantum and
classical ones. To this extent, let us define the \emph{Husimi function} of
any quantum state $\rho_{\varepsilon}$ as
\begin{equation}
  \label{eq: Husimi}
  f_\varepsilon (z) := \left\langle z_\varepsilon\left|\rho_\varepsilon\right|z_{\varepsilon}\right\rangle_{\mathcal{H}_{\varepsilon}}\;,
\end{equation}
where, for any $ z \in \mathcal{H} $, $ \left|z_\varepsilon \right\rangle $ stands for the coherent state
centered at $ z $ (see \eqref{eq: coherent} below). It is well-known that the
Husimi function is a probability density on $\mathcal{H}$ when multiplied by
$(\pi\varepsilon)^{-d}$, and therefore we can define the \emph{Wehrl entropy} as follows:
let $\mathrm{d}\varphi_{\varepsilon}(z)= \frac{1}{(\pi\varepsilon)^d}f_{\varepsilon}(z)\mathrm{d}z$ be the
probability measure associated to the Husimi function, then,
\begin{equation}
  \label{eq:W}
  S_{\mathrm{W},\varepsilon} (\rho_\varepsilon)  := S_{\mathrm{B}}(\varphi_{\varepsilon}) - d\log (\pi\varepsilon) =- \int_{\mathcal{H}} f_\varepsilon (z) \log (f_\varepsilon(z)) \frac{\mathrm{d} z}{(\pi \varepsilon)^d}\;.
\end{equation}
The corresponding \emph{Wehrl relative entropy} is thus
\begin{equation}
  \label{eq:rW}
  S_{\mathrm{W},\varepsilon}(\rho_{\varepsilon}\Vert \sigma_{\varepsilon})= S_{\mathrm{B}}(\varphi_{\varepsilon} \Vert \phi_{\varepsilon})= \int_{\mathcal{H}}^{}f_{\varepsilon}(z)\Bigl(\log\bigl(f_{\varepsilon}(z)\bigr)- \log\bigl(g_{\varepsilon}(z)\bigr) \Bigr)  \frac{\mathrm{d}z}{(\pi\varepsilon)^d}\;,
\end{equation}
where $g_{\varepsilon}$ is the Husimi function of $\sigma_{\varepsilon}$, and
$\phi_{\varepsilon}$ its associated probability measure.

\subsection{Classical and Quantum Free Energies}
\label{sec:class-quant-free}

A key physical quantity associated to classical and quantum systems in
statical mechanics is the {\it (Helmholtz) free energy}, which allows to
identify the (Gibbs) equilibrium states via the variational principle, as the
minimizers of such a quantity\footnote{We refer to such problem as the
  \emph{Gibbs variational problem}.}. The classical and quantum free energies
at inverse temperature $\beta$ are given respectively by
\begin{equation}
  F_{\mathrm{B},\beta}(\mu) = \int_{\mathcal{H}} h(z) \mathrm{d} \mu(z) - \tfrac{1}{\beta} S_{\mathrm{B}}(\mu),	
\end{equation} 
and
\begin{equation}
  F_{\diamond,\beta,\varepsilon}(\rho_{\varepsilon}) = \Tr_{\mathcal{H}_{\varepsilon}} (H_{\varepsilon} \rho_\varepsilon )  - \tfrac{1}{\beta} S_{\diamond,\varepsilon}(\rho_\varepsilon),	
\end{equation}
where $\diamond$ stands either for $\mathrm{vN}$ (von Neumann) or
$\mathrm{W}$ (Wehrl), depending on the choice of quantum entropy. As such,
the classical free energy is defined only for measures that are absolutely
continuous with respect to the Lebesgue measure; however, classical and
quantum free energies can be rewritten resorting to Gibbs states and
\emph{relative entropies} (and thus the classical free energy makes sense for
any measure). Given the classical Gibbs state
\begin{equation*}
  \mathrm{d}\gamma_{\beta}(z)=\frac{1}{Z_{\beta,0}}\, e^{-\beta h(z)}\mathrm{d}z\;,
\end{equation*}
it is not difficult to verify that
\begin{equation}
  \label{eq:4}
  F_{\mathrm{B},\beta}(\mu)= \tfrac{1}{\beta}S_{\mathrm{B}}(\mu \Vert \gamma_{\beta}) + F_{\mathrm{B},\beta}(\gamma_{\beta})=\tfrac{1}{\beta}\bigl(S_{\mathrm{B}}(\mu \Vert \gamma_{\beta})- \log Z_{\beta,0}\bigr) \;,
\end{equation}
and this extends to any probability measure $\mu$ on
$\mathcal{H}$. Furthermore, since the relative entropy is always non-negative
-- vanishing if and only if $\mu=\gamma_{\beta}$ -- it follows that the free
energy is bounded from below and has the Gibbs state as the unique
minimizer. As a matter of fact, it is convenient to consider directly the
``relative'' free energy
\begin{equation}
  \label{eq:5}
  \widetilde{F}_{\mathrm{B},\beta}(\mu)= \tfrac{1}{\beta}S_{\mathrm{B}}(\mu \Vert \gamma_{\beta}) 
\end{equation}
for the minimization problem. Analogously, denoting
\begin{equation*}
  \Gamma_{\beta,\varepsilon}= Z_{\beta,\varepsilon}^{-1} e^{-\beta H_{\varepsilon}}
\end{equation*}
the quantum Gibbs state, the von Neumann free energy can be written as
\begin{equation}
  \label{eq:6}
  F_{\mathrm{vN},\beta,\varepsilon}(\rho_{\varepsilon}) = \tfrac{1}{\beta} S_{\mathrm{vN},\varepsilon}(\rho_{\varepsilon}\Vert \Gamma_{\beta,\varepsilon}) + F_{\mathrm{vN},\beta,\varepsilon}(\Gamma_{\beta,\varepsilon})= \tfrac{1}{\beta}\bigl(S_{\mathrm{vN},\varepsilon}(\rho_{\varepsilon}\Vert \Gamma_{\beta,\varepsilon})- \log Z_{\beta,\varepsilon}\bigr)\;,
\end{equation}
from which it follows that it is bounded from below and minimized by
$\Gamma_{\beta,\varepsilon}$, and we can write the relative quantum free
energy
\begin{equation}
  \label{eq:7}
  \widetilde{F}_{\mathrm{vN},\beta,\varepsilon}(\rho_{\varepsilon}) = \tfrac{1}{\beta} S_{\mathrm{vN},\varepsilon}(\rho_{\varepsilon}\Vert \Gamma_{\beta,\varepsilon})\;.
\end{equation}
Finally, let us focus on the Wehrl free energy. Let us denote by
$h_{\varepsilon}^{\mathrm{up}}(z)$ the so-called upper symbol of
$H_{\varepsilon}$: the function on $\mathcal{H}$ (supposing it exists) such
that
\begin{equation*}
  H_{\varepsilon}=\int_{\mathcal{H}}^{}h^{\mathrm{up}}_{\varepsilon}(z) \left\lvert z_{\varepsilon}\right\rangle\left\langle z_{\varepsilon}\right\rvert \frac{\mathrm{d}z}{(\pi\varepsilon)^d}\;.
\end{equation*}
Then,
\begin{multline}
  \label{eq:8}
  F_{W,\beta,\varepsilon}(\rho_{\varepsilon})= \int_{\mathcal{H}}^{}h_{\varepsilon}^{\mathrm{up}}(z) f_{\varepsilon}(z)  \frac{\mathrm{d}z}{(\pi\varepsilon)^d} - \tfrac{1}{\beta}S_{\mathrm{B}}(\varphi_{\varepsilon}) +\tfrac{d}{\beta}\log (\pi\varepsilon) \\= \tfrac{1}{\beta}S_{\mathrm{B}}(\varphi_{\varepsilon}\Vert \gamma_{\beta,\varepsilon}) + F_{\mathrm{B},\beta}(\gamma_{\beta,\varepsilon}) +\tfrac{d}{\beta}\log (\pi\varepsilon)\;,
\end{multline}
where
\begin{equation*}
  \mathrm{d}\gamma_{\beta,\varepsilon}(z)= \frac{e^{-\beta h_{\varepsilon}^{\mathrm{up}}(z)}\mathrm{d}z}{\int_{\mathcal{H}}^{}e^{-\beta h_{\varepsilon}^{\mathrm{up}}(z)} \mathrm{d}z}
\end{equation*}
is the classical Gibbs state for the upper symbol of $H_{\varepsilon}$. Again, this
shows, since the upper symbol is bounded from below, that the Wehrl free
energy is bounded from below. Its minimizers are the quantum states -- if any
-- whose Husimi function is
\begin{equation*}
  (\pi\varepsilon)^d\frac{e^{-\beta h_{\varepsilon}^{\mathrm{up}}(z)}}{\int_{\mathcal{H}}^{}e^{-\beta h_{\varepsilon}^{\mathrm{up}}(z)} \mathrm{d}z} \; .
\end{equation*}
The relative Wehrl free energy can be written as
\begin{equation}
  \label{eq:1}
  \widetilde{F}_{W,\beta,\varepsilon}(\rho_{\varepsilon})= \tfrac{1}{\beta} S_{\mathrm{B}}(\varphi_{\varepsilon} \Vert \gamma_{\beta,\varepsilon})\;.
\end{equation}

\subsection{Symbols and Quantization}
\label{sec:symbols-quantization}

In order to present our main results, let us introduce the Fock
representation $\mathcal{H}_{\varepsilon}= \mathcal{F}_\varepsilon(\mathcal{H})$ of the
canonical commutation relations, and the quantization maps from functions on
$\mathcal{H}$ to linear operators on $\mathcal{F}_\varepsilon(\mathcal{H})$.

As discussed above, we take the classical phase space $\mathcal{H}$ to be a complex
finite dimensional Hilbert space, naturally identifiable with the complex
structure of $\mathbb{R}^{2 d}$ endowed with the Euclidean norm and canonical
symplectic form; let us denote by $ \left\langle z|w \right\rangle_{\mathcal{H}} $ the inner
product on $\mathcal{H}$. The bosonic, or symmetric, Fock space over (a generic Hilbert
space) $\mathcal{H}$ is constructed as the completion of the symmetric tensor algebra
over $\mathcal{H}$:
\begin{equation}
  \mathcal{F}_\varepsilon (\mathcal{H}) =   \bigoplus_{n=0}^{\infty} \mathcal{H}^{\otimes_{\mathrm{s}} n}\;.
\end{equation}
Here, $ \mathcal{H}^{\otimes_{\mathrm{s}} 0} = \mathbb{C} $ and we denote by
$ \left|\Omega \right\rangle := (1,0,\dotsc) $ the {\it vacuum}
vector. Quantum mechanically, $ \mathcal{H}^{\otimes_{\mathrm{s}} n} $ is
meant to describe $ n $ indistinguishable particles (or excitations) and,
therefore, the number of particles in $ \mathcal{F}_\varepsilon(\mathcal{H}) $ is allowed
to vary. The creation and annihilation operators $a_\varepsilon^*(z)$ and
$a_\varepsilon(z)$ are defined on $\mathcal{F}_\varepsilon (\mathcal{H})$ for each $z \in
\mathcal{H}$ as\footnote{We denote here by $ \widehat{\cdot} $ a missing term
  in the expression.}
\begin{equation}
  a_\varepsilon^* (z)  \left|\psi_n \right\rangle  = \sqrt{\varepsilon (n+1)} \left|z \right\rangle \otimes_{\mathrm{s}} \left|\psi_n \right\rangle, \qquad	\text{for } \psi_n \in \mathcal{H}^{\otimes_{\mathrm{s}} n}\;,  
\end{equation}
\begin{equation}
  a_\varepsilon (z) \left|\varphi_1 \right\rangle \otimes_{\mathrm{s}} \cdots \otimes_{\mathrm{s}} \left|\varphi_n \right\rangle = \sqrt{\varepsilon} \sum_{i=1}^n  \left\langle z|\varphi_i \right\rangle \left|\varphi_1 \right\rangle \otimes_{\mathrm{s}} \cdots \otimes_{\mathrm{s}} \widehat{\left|\varphi_i \right\rangle} \otimes_{\mathrm{s}} \cdots \otimes_{\mathrm{s}} \left|\varphi_n \right\rangle\;,
\end{equation}
for $ \varphi_i \in \mathcal{H} $, $ i = 1, \ldots, n $. Note that $a_\varepsilon^*(z)$
(resp.\ $a_\varepsilon(z)$) maps the $n$-particle space to the $n+1$ (resp.\ $n-1$)
particle space, while $ \mathcal{H}^{\otimes_{\mathrm{s}} 0} $ and the vacuum
$ \left|\Omega \right\rangle $ are annihilated by $ a_\varepsilon(z) $, \emph{i.e.}, $ a_\varepsilon(z)
\mathcal{H}^{\otimes_{\mathrm{s}} 0} = \{ 0 \} $.  These operators are
unbounded but closable and their closures, still denoted by the same symbol,
are adjoint one other. The $\varepsilon$-dependent canonical commutation
relations (CCRs) read
\begin{equation}
  \tag{CCR} \label{CCR}
  \begin{aligned}
    \relax \left[a_\varepsilon(z),a_\varepsilon^*(w) \right] &= \varepsilon \left\langle z|w \right\rangle_\mathcal{H} \mathds{1}_{\mathcal{F}_\varepsilon (\mathcal{H})}\;, \\
    \left[a_\varepsilon(z),a_\varepsilon(w)\right] &= \left[a_\varepsilon^*(z),a_\varepsilon^*(w)\right] = 0\;.
  \end{aligned}
\end{equation}
Given an orthonormal basis $\{\phi_j\}$ of $\mathcal{H}$, let us define
\begin{gather*}
  a_{\varepsilon,j}:= a_{\varepsilon}(\phi_j)\;,\\
  a^{*}_{\varepsilon,j}:= a^{*}_{\varepsilon}(\phi_j)\;.
\end{gather*}

Given a self-adjoint operator $A$ with domain $D(A)$, its second quantization
$\mathrm{d} \Gamma_{\varepsilon} (A)$ is defined on $\otimes_{n \in
  \mathbb{N}}^{\mathrm{alg}} D(A)^{\otimes_s n}$ as
\begin{equation}
  \left. \mathrm{d} \Gamma_{\varepsilon} (A) \right|_{D(A)^{\otimes_{\mathrm{s}} n}} = \varepsilon \sum_{k=1}^n \mathds{1} \otimes \cdots \otimes\overset{k\text{-th term}}{\overbrace{\, A \,}} \otimes \cdots \otimes \mathds{1}\;.
\end{equation}
Among the linear operators on $ \mathcal{F}_\varepsilon(\mathcal{H}) $, a very important
role is played by the {\it Weyl operator}
\begin{equation}
  W_\varepsilon(z) : = e^{ \frac{i}{\sqrt{2}}    (a_\varepsilon(z)+a_\varepsilon^*(z)  )},	\qquad	z \in \mathcal{H}\;,
\end{equation}
which is involved in the definition of {\it coherent states}: for any
$z \in \mathcal{H}$, we set
\begin{equation}
  \label{eq: coherent}
  \left|z_\varepsilon \right\rangle := W_\varepsilon \left( \tfrac{\sqrt{2} z }{i \varepsilon} \right) \left|\Omega \right\rangle = e^{ \frac{1}{\sqrt{\varepsilon}} ( a_\varepsilon^*(z) - a_\varepsilon(z)) } \left|\Omega \right\rangle\;.
\end{equation}
Coherent states are very well studied, and enjoy several notable properties;
let us mention that they have minimal uncertainty -- \emph{i.e.}, they can be
considered the ``most classical'' quantum states-- and they are eigenvectors
of the annihilation operators:
\begin{equation}
  \label{eq:eigenstate}
  a_\varepsilon(z) \left|w_\varepsilon \right\rangle =  \left\langle z|w \right\rangle_\mathcal{H} \left|w_\varepsilon \right\rangle\;.
\end{equation}
Furthermore, the family $\{ \left|z_\varepsilon \right\rangle \}_{z \in
  \mathcal{H}} $ forms an overcomplete set whenever $\mathcal{H}$ is finite
dimensional, with dimension $d$:
\begin{equation}
  \label{eq:overcomplete}
  \int_\mathcal{H} \left|z_\varepsilon \right\rangle \left\langle z_{\varepsilon}\right| \frac{\mathrm{d} z}{(\pi \varepsilon)^d} = \mathds{1}_\mathcal{H}\;.
\end{equation}
The number operator $N_{\varepsilon}=\mathrm{d}\Gamma_{\varepsilon}(1)$ is positive, and counts the
number of particles in each sector. It satisfies the following relations with
respect to the creation and annihilation operators: for any measurable
function $f:\mathbb{R}_+\to \mathbb{C}$ and $\psi\in D(N_{\varepsilon}^{1/2})$,  
\begin{gather}
  \label{eq:9}
  f(N_{\varepsilon})a_{\varepsilon}(z)= a_{\varepsilon}(z)f(N_{\varepsilon}-\varepsilon)\;,\; f(N_{\varepsilon})a^{*}_{\varepsilon}(z)= a^{*}_{\varepsilon}(z)f(N_{\varepsilon}+\varepsilon)\;;\\
  \label{eq:10}
  \lVert a_{\varepsilon}(z)\psi  \rVert_{}^{}\leq \lVert z  \rVert_{\mathcal{H}}^{}\lVert N_{\varepsilon}^{1/2}\psi  \rVert_{}^{}\;,\; \lVert a^{*}_{\varepsilon}(z)\psi  \rVert_{}^{}\leq \lVert z  \rVert_{\mathcal{H}}^{}\lVert N_{\varepsilon}^{1/2}\psi  \rVert_{}^{} + \varepsilon\lVert z  \rVert_{\mathcal{H}}^{}\lVert \psi  \rVert_{}^{}\;.
\end{gather}

Another important ingredient of our analysis is {\it Wick quantization},
which provides a map from {\it symbols}, \emph{i.e.}, suitable functions on $
\mathcal{H} $ to operators on $ \mathcal{F}_\varepsilon(\mathcal{H}) $. Formally, it can be easily defined on
homogeneous polynomials by simply replacing in each monomial $ \left\langle z|f
\right\rangle_{\mathcal{H}} $ (resp.\ $ \left\langle f|z \right\rangle_{\mathcal{H}} $) by $ a_\varepsilon^* (f) $ (resp.\ $
a_\varepsilon(f) $), and writing the expression in {\it normal order}, \emph{i.e.},
with all the creation operators to the left and all the annihilation
operators to the right. Let us recall here only a few elementary properties
of the Wick quantization and refer to, \emph{e.g.},
\cite{AmmariNier2008,Berezin1991} for a complete discussion.

Given $p,q \in \mathbb{N}$, define the set of $(p,q)-$homogeneous polynomial symbols of
degree $p+q$ on $\mathcal{H}$ as
\begin{equation}
  \mathcal{P}_{p,q} = \left\{ \left. b(z) = \left\langle z^{\otimes q}\left|\widetilde{b} z^{\otimes p}\right. \right\rangle_{\mathcal{H}^{\otimes q} } \: \right| \: \widetilde{b} \in \mathcal{L} (\mathcal{H}^{\otimes p}; \mathcal{H}^{\otimes q})\right\}\;.
\end{equation}
Note that for any $ b \in \mathcal{P}_{p,q} $ we can recover $\widetilde{b}$ via
Gâteaux
differentiation.

\begin{definition}[Wick quantization]
  \label[definition]{defn: wick}
  \mbox{}	\\
  For any symbol $b \in \mathcal{P}_{p,q}$, the operator $ b^{\mathrm{Wick}}
  \in \mathcal{L}(\mathcal{F}_\varepsilon^{\mathrm{fin}}) $, with $
  \mathcal{F}_\varepsilon^{\mathrm{fin}} : = \oplus_{n=0}^{\mathrm{alg}}
  \mathcal{H}^{\otimes_{\mathrm{s}} n} $, is given by
  \begin{equation}
    \label{eq:Wick}
    b^{\mathrm{Wick}}= \sum_{i_1,\dotsc,i_q;j_1,\dotsc,j_p=1}^{d} \widetilde{b}_{i_1,\dotsc,i_q;j_1,\dotsc,j_p} \:a^{*}_{\varepsilon,i_1}\dotsm a^{*}_{\varepsilon,i_q} a_{\varepsilon,j_1}\dotsm a_{\varepsilon,j_p}\;,
  \end{equation}
  where
  \begin{equation*}
    \widetilde{b}_{i_1,\dotsc,i_q;j_1,\dotsc,j_p} = \langle \phi_{i_1}\otimes \dotsm\otimes \phi_{i_q}  \:\vert\: \widetilde{b}\: \phi_{j_1}\dotsm \phi_{j_p} \rangle_{\mathcal{H}^{\otimes_{} q}}\;.
  \end{equation*}
\end{definition}
Such operators are closable, and we still denote by $b^{\mathrm{Wick}}$ their
closure. Some key examples are given by
$$
\left\langle \xi|z \right\rangle^{\mathrm{Wick}} = a_{\varepsilon} (\xi),
\qquad \left\langle z|\xi \right\rangle^{\mathrm{Wick}} = a_{\varepsilon}^*
(\xi), \qquad \left\langle z\right|A\left|z \right\rangle^{\mathrm{Wick}} =
\mathrm{d} \Gamma_{\varepsilon} (A)\;,
$$
for any fixed $ \xi \in \mathcal{H} $ and $A \in \mathcal{L}(\mathcal{H})
$. Setting $ A = \mathds{1}_{\mathcal{H}}$, we recover the {\it number
  operator}
\begin{equation}
  \label{eq:Neps}
  N_{\varepsilon} : = \left( \left\| z \right\|^2 \right)^{\mathrm{Wick}} = \mathrm{d} \Gamma_{\varepsilon}(1)\;.
\end{equation}

Analogously, one can introduce the {\it anti-Wick quantization}, which
formally consists of exchanging the role of $ a_\varepsilon^* $ and
$a_\varepsilon $ in the procedure.
\begin{definition}[anti-Wick quantization]
  \mbox{}	\\
  For any (suitably regular) symbol $ b: \mathcal{H}\to \mathbb{C}$, the
  operator $ b^{\mathrm{a-Wick}} $ is densely defined and closable on
  $\mathcal{F}_\varepsilon(\mathcal{H})$, and has the form
  \begin{equation}
    b^{\mathrm{a-Wick}} = \int_{\mathcal{H}} b(z)  \left|z_\varepsilon \right\rangle\left\langle z_\varepsilon \right| \: \frac{\mathrm{d} z}{(\pi \varepsilon)^d}\;. 
  \end{equation}
\end{definition}
It is not difficult to verify that the above definition satisfies the
heuristics mentioned above, \emph{i.e.}, it formally transforms $
\left\langle z|f \right\rangle_{\mathcal{H}} $ and $ \left\langle f|z
\right\rangle_{\mathcal{H}} $ into $ a_\varepsilon (f) $ and $
a^*_\varepsilon(f) $, respectively, and moving all the $ a $'s to the left
and all the $a^*$'s to the right (anti-normal ordering): recalling
\eqref{eq:eigenstate} and \eqref{eq:overcomplete}, we indeed get
\begin{multline}
  \left(\left\langle f|z \right\rangle_{\mathcal{H}} \left\langle z|g \right\rangle_{\mathcal{H}}  \right)^{\mathrm{a-Wick}} =
  \int_{\mathcal{H}} \overline{\left\langle z|f \right\rangle_{\mathcal{H}}} \left\langle z|g \right\rangle_{\mathcal{H}} | z_\varepsilon \rangle \langle z_\varepsilon | \frac{\mathrm{d} z}{(\pi \varepsilon)^d} \\
  = \int_{\mathcal{H}} a_\varepsilon(g) | z_\varepsilon \rangle \langle z_\varepsilon | a_\varepsilon^*(f)\frac{\mathrm{d} z}{(\pi \varepsilon)^d} = a_\varepsilon(g) a_\varepsilon^*(f)\;.
\end{multline}

We are now in a position to state the assumptions on the Hamiltonian
$H_{\varepsilon}$ of the system, needed for our results. We assume that
$H_{\varepsilon}= h^{\mathrm{Wick}}$, with $h$ belonging to a suitable symbol
class $\mathcal{S}$.

\begin{definition}[Symbol class $\mathcal{S}$] 
  \mbox{}	\\
  A polynomial symbol $h\in \oplus^{\mathrm{alg}}_{p,q} \mathcal{P}_{p,q}$ belongs to the
  class $\mathcal{S}$ if and only if
  \begin{itemize}
  \item For all $z\in \mathcal{H}$, $h(z)= \overline{h(z)}$;
  \item $h(z)= h_0(z)+ V(z)$, with
    \begin{equation*}
      h_0(z)= \sum_{p=1}^{p_{\mathrm{max}}} \langle z^{\otimes p}  \vert \widetilde{h}_{p} \,z^{\otimes p}  \rangle_{\mathcal{H}^{\otimes p}}
    \end{equation*}
    for some $p_{\mathrm{max}}\geq 1$ and $\widetilde{h}_p\geq 0$ (bounded)
    operator on $\mathcal{H}^{\otimes p}$, $\widetilde{h}_{p_{\mathrm{max}}}>0$, and $V(z)$
    of maximal degree strictly smaller than $2p_{\mathrm{max}}$.
  \end{itemize}
\end{definition}

For any $ h \in \mathcal{S} $, we denote its Wick quantization by
\begin{equation}
  H_\varepsilon := h^{\mathrm{Wick}}\;.
\end{equation}
By a standard application of \eqref{eq:9} and \eqref{eq:10} together with
Kato-Rellich's theorem, there exist constants $C>0$ and $\widetilde{C}\geq 0$
such that
\begin{equation*}
  H_{\varepsilon}\geq C \bigl(\lVert z  \rVert_{}^{2p_{\mathrm{max}}}\bigr)^{\mathrm{Wick}} - \widetilde{C}\;.
\end{equation*}
Noticing that, by definition,
\begin{equation}
  h(z) = \left\langle z_\varepsilon\left|H_\varepsilon\right|z_\varepsilon\right\rangle_{\mathcal{F}_\varepsilon(\mathcal{H})}\;,
\end{equation}
the bound above yields
\begin{equation}
  \label{eq:symbol lb}
  h(z) \geq C \left\| z \right\|^{2p_{\mathrm{max}}} - \widetilde{C}.
\end{equation}


\begin{remark}[Partition functions]
  \mbox{}	\\
  Since $h$ grows at least like $\left\| z \right\|^{2p_{\mathrm{max}}}$,
  \begin{equation}
    Z_{\beta,0} =  \int_{\mathcal{H}}  e^{- \beta h(z)} \mathrm{d} z < + \infty.
  \end{equation}
  Also, as shown in \cite[Prop. 5.2.27]{BratelliRobinson2013}, it holds that
  \begin{equation}
    Z_{\beta,\varepsilon} = \Tr_{\mathcal{F}_\varepsilon(\mathcal{H})} \left( e^{-\beta H_{\varepsilon}} \right) < +\infty.
  \end{equation}
\end{remark}

\subsection{Convergence of Gibbs Entropies}
\label{sec:conv-entr}

Our first result concerns the limit of quantum entropies of Gibbs states
$\Gamma_{\beta,\varepsilon}$, for a suitable class of Hamiltonians
$H_{\varepsilon}= h^{\mathrm{Wick}} $. Let us immediately point out, however,
that the quantum entropies diverge in the semiclassical limit $ \varepsilon
\to 0 $. It is possible to gather an intuition by looking at the partition
functions: taking, \emph{e.g.}, the harmonic oscillator Hamiltonian on $
\mathbb{R}^d $, one immediately verifies that $ Z_{\beta, \varepsilon} $
defined in \eqref{eq: zbe} contains a diverging term of order
$\varepsilon^{-d}$. Therefore, in order to prove convergence, it is necessary
to subtract the corresponding divergent expressions from the entropies,
namely $ - d \log (\pi \varepsilon) $. Note that the divergent term only
depends on the dimension of the phase space $\mathcal{H}$, and not on the
choice of the Hamiltonian.

In order to state our result, we need another assumption on the Hamiltonian
$h $, or, more precisely, on the quantum Gibs state associated to its Wick
quantization.

\begin{assumption*}
  \mbox{}	\\
  The family of Gibbs states $(\Gamma_{\beta,\varepsilon}=(Z_{\beta,
    \varepsilon})^{-1}e^{-\beta H_\varepsilon})_{\varepsilon \in (0,1]} \in
  \mathfrak{S}_+^1(\mathcal{F}_\varepsilon)$ satisfies \eqref{A}, if $\forall k\in
  \mathbb{N}, \exists C_k<+\infty$ independent of $ \varepsilon $ such that
  \begin{equation} \tag{A} \label{A} \left\lVert
      (N_\varepsilon+\varepsilon)^{\frac{k}{2}} e^{-\beta H_\varepsilon}
      (N_\varepsilon+\varepsilon)^{\frac{k}{2}}\right\rVert_{\mathcal L
      (\mathcal{F}_\varepsilon (\mathcal{H}))} \leq C_k\;.
  \end{equation}
\end{assumption*}

\begin{remark}[Role of \eqref{A}]
  \mbox{}	\\
  The assumption \eqref{A} will be used to prove that the Husimi function
  \eqref{eq: Husimi} associated to the Gibbs state decays fast enough at
  infinity with appropriate uniform bounds. A simple case in which \eqref{A}
  trivially holds true is when the Gibbs state is constructed from from a
  Hamiltonian $ H_{\varepsilon} $ commuting with the number operator, but in
  this case one can directly prove that the Husimi function has an
  exponential decay with uniform rate in $\varepsilon$.
\end{remark}

Under this assumption, we prove that both the von Neumann and Wehrl entropies
of the Gibbs state converge to the Boltzmann entropy of its classical
counterpart.

\begin{thm}[Convergence of entropies]
  \label[theorem]{thm:convergence}
  \mbox{}	\\
  Let $h \in \mathcal{S} $ and let $H_\varepsilon = h^{\mathrm{Wick}} $ be
  its Wick quantization. Let also the Gibbs state $\Gamma_{\beta,\varepsilon}
  := (Z_{\beta, \varepsilon})^{-1} e^{-\beta H_\varepsilon}$ satisfy the
  assumption \eqref{A}. 
  Then, Gibbs entropies converge:
  \begin{equation}
    \label{cvg entropy}
    S_{\diamond,\varepsilon} (\Gamma_{\beta,\varepsilon}) + d \log( \pi \varepsilon) \xrightarrow[\varepsilon \to 0]{} S_{\mathrm{B}} ( \gamma_{\beta} )\;,
  \end{equation}
  where $\diamond$ stands either for $\mathrm{vN}$ (von Neumann) or
  $\mathrm{W}$ (Wehrl).
\end{thm}

\begin{remark}[Entropy rescaling]
  \mbox{}	\\
  The term $d \log (\pi \varepsilon)$ acts as a renormalization or rescaling
  of the quantum entropy and is independent of all the details of the system
  but the dimension $ d $ of the classical phase space $ \mathcal{H} $,
  \emph{i.e.}, the classical number of degrees of freedom.
\end{remark}

As a byproduct of this theorem -- in particular combining it with \cref{cor:2}
that states that the classical Gibbs measure is the unique Wigner measure of
the quantum Gibbs state -- we also prove convergence of both Wehrl and von
Neumann Gibbs free energies.
\begin{corollary}[Convergence of the Gibbs free energies]
  \label[corollary]{cor:3}
  \mbox{}\\
  \begin{equation*}
    \lim_{\varepsilon\to 0} \Bigl[F_{\diamond, \beta,\varepsilon}(\Gamma_{\beta,\varepsilon}) - \tfrac{d}{\beta}\log(\pi\varepsilon)\Bigr]= F_{\mathrm{B},\beta}(\gamma_{\beta})\;.
  \end{equation*}
\end{corollary}

\subsection{$\Gamma$-convergence of Free Energies}
\label{sec:gamma-conv-free}

The other main result of the paper concerns the $\Gamma$-convergence of the
quantum free energy functionals in the classical limit. The
$\Gamma$-convergence of functionals is a powerful tool in calculus of
variations (see \cite{Braides} for a gentle but detailed introduction to the
topic), and it seems very well tailored to study the Gibbs variational
problem in the classical limit.

Firstly, we have to define precisely the free energy functionals, both
classical and quantum, paying attention to their respective domains. Let us
start from the latter: to prove a $ \Gamma-$convergence result, it is necessary to
identify a common domain for both the quantum functionals and their classical
limit; this is done by taking \emph{Fourier transforms}. The classical states
are in bijection, by means of Bochner's theorem, with the continuous and
positive definite complex functions on the dual $\mathcal{H}^{*}\cong \mathcal{H}$. More precisely, a
function $G_0:\mathcal{H}^{*}\to \mathbb{C}$ is the characteristic function (Fourier transform) of
a probability measure if and only if the following properties are satisfied:
\begin{itemize}
\item $G_0(0)=1$;
\item $G_0$ is continuous;
\item for any finite collections $\{\zeta_j\}_{j=1}^N\subset \mathbb{C}$ and $\{z_j\}_{j=1}^N\subset
  \mathcal{H}$,
  \begin{equation*}
    \sum_{j,k=1}^N \bar{\zeta}_k\zeta_j G_0(z_j-z_k)\geq 0\;.
  \end{equation*}
\end{itemize}
Let us denote by $\mathfrak{G}_0$ the set of such continuous and positive definite
functions. Analogously, quantum density matrices are in bijection, by
noncommutative Bochner's theorem, with the continuous and \emph{quantum}
positive definite complex functions on $\mathcal{H}^{*}$. More precisely, a function
$G_{\varepsilon}: \mathcal{H}^{*}\to \mathbb{C}$ is the noncommutative characteristic function of a density
matrix (normal Schrödinger state) for the $\varepsilon$ canonical commutation relations
\eqref{CCR} if and only if:
\begin{itemize}
\item $G_\varepsilon(0)=1$;
\item $G_\varepsilon$ is continuous;
\item for any finite collections $\{\zeta_j\}_{j=1}^N\subset \mathbb{C}$ and $\{z_j\}_{j=1}^N\subset
  \mathcal{H}$,
  \begin{equation*}
    \sum_{j,k=1}^N \bar{\zeta}_k\zeta_j G_\varepsilon(z_j-z_k) e^{-i \varepsilon \mathrm{Im}\langle  z_k \vert z_j \rangle_{\mathcal{H}}}\geq 0\;.
  \end{equation*}
\end{itemize}
Let us denote by $\mathfrak{G}_{\varepsilon}$ the set of continuous and $\varepsilon$-quantum positive
definite functions.

Now that quantum and classical state are put on the same grounds, as
complex-valued continuous functions on $\mathcal{H}^{*}$, we can define the topological
space $(X,\mathcal{T})$ to be the domain of the $\Gamma$--functionals.
\begin{definition}[$(X,\mathcal{T})$]
  \label[definition]{def:1}
  We define the subspace $X\subset \mathcal{C}(\mathcal{H}^{*};\mathbb{C})$ by
  \begin{equation*}
    X= \mathfrak{G}_0 \cup \Bigl(\bigcup_{n\in \mathbb{N}^{*}} \mathfrak{G}_{\frac{1}{n}}\Bigr)\;.
  \end{equation*}
  We endow $X$ with the subspace topology $\mathcal{T}$ of the compact-open topology on
  $\mathcal{C}(\mathcal{H}^{*};\mathbb{C})$.
\end{definition}
\begin{remark}
  \label{rem:3}
  A couple of observations are in order:
  \begin{itemize}
  \item The intersection $\mathfrak{G}_0\cap \mathfrak{G}_{\varepsilon}$ is not empty for any $\varepsilon>0$, Gaussian
    functions being the most prominent example. Therefore, the Fourier
    transform is a bijection on $\mathfrak{G}_0$ and any $\mathfrak{G}_{\varepsilon}$ separately, but
    \emph{it is not such} on the union, where it fails to be
    injective. Therefore, we cannot directly take $X$ to be the union of
    probability measures and density matrices (with the initial topology
    making the Fourier transforms continuous when the image is endowed with
    the compact-open topology): such topological space would fail to separate
    points; we thus need to ``inverse Fourier transform'' correctly in
    defining the quantum and classical functionals, as we do here below.
  \item An $\varepsilon$-quantum positive definite function defines a quantum state for
    the $\varepsilon$-CCR only, \emph{it does not} define a state for an $\varepsilon'$-CCR,
    whenever $\varepsilon'\neq \varepsilon$.
  \end{itemize}
\end{remark}

To define the relative free energy functionals as functionals on $X$, we
introduce the following useful notation. Let $x\in \mathfrak{G}_{\varepsilon}$, $\varepsilon\geq 0$, we denote by
$\check{x}_{\varepsilon}$ the $\varepsilon$-inverse Fourier transform of $x$: $\check{x}_{\varepsilon}$ is
the unique density matrix ($\varepsilon> 0$) or probability measure ($\varepsilon=0$) whose
Fourier transform is $x$. If no confusion arises, we might omit the subscript
$_{\varepsilon}$ from $\check{x}$. This leads to the following definition:
\begin{equation}
  \label{eq: fbf}
  \widetilde{F}_{\mathrm{B},\beta} (x) :=
  \begin{cases} 
    \; \widetilde{F}_{\mathrm{B},\beta}(\check{x}_0)   &	\text{if }  x\in \mathfrak{G}_0	\\
    \; +\infty & \text{otherwise}
  \end{cases}\;.
\end{equation}
Analogously, for the von Neumann and Wehrl free energies,
\begin{equation}
  \label{eq: fwf}
  \widetilde{F}_{\diamond,\beta,\varepsilon}(x):=
  \begin{cases} 
    \; \widetilde{F}_{\diamond,\beta,\varepsilon}(\check{x}_{\varepsilon})   &	\text{if }  x \in \mathfrak{G}_{\varepsilon}\\
    \;+\infty & \text{otherwise}
  \end{cases}\;;
\end{equation}
where as before $\diamond$ stands either for $\mathrm{vN}$ or $\mathrm{W}$.

Let us now recall the sequential definition of $ \Gamma-$convergence in
first-countable spaces. For a more general definition and overview of
$\Gamma-$convergence, refer again to \cite{Braides}.

\begin{definition}[$\Gamma-$convergence]
  \label[definition]{def: Gamma}
  \mbox{}	\\
  Let $X$ be a first-countable topological space. Given a sequence of
  functionals $ \left\{ F_n \right\}_{n \in \mathbb{N}} $ and $ F $ such that
  $F_n , F : X \rightarrow \overline{\mathbb{R}} = \mathbb{R} \cup \{+\infty\}
  $, we say that $F_n \xrightarrow[n \to \infty] {\Gamma} F $ if the
  following holds:
  \begin{itemize}
  \item \emph{($\Gamma-$lower bound)} $ \forall x \in X, \forall (x_n)_{n\in \mathbb{N}}
    \subset X $ such that $ x_n \xrightarrow[n \to \infty]{} x $,
    \begin{displaymath}
      F(x) \leq \liminf_{n \to \infty} F_n (x_n);
    \end{displaymath}	
  \item \emph{($\Gamma-$upper bound)} $ \forall x \in X, \exists (x_n)_{n\in \mathbb{N}}
    \subset X $ such that $ x_n \xrightarrow[n \to \infty]{} x $ and
    \begin{displaymath}
      \limsup_{n \to \infty} F_n (x_n) \leq F(x).
    \end{displaymath}
  \end{itemize}
\end{definition}

Let us remark that for any convergent sequence $x_n\to x$ such that along a
subsequence $n_k\to \infty$ we have $x_{n_k}\in \mathfrak{G}_{\frac{1}{n_k}}$ for all $k\in \mathbb{N}$, it
holds that $x\in \mathfrak{G}_0$; furthermore, convergence along this subsequence yields
convergence of $\check{x}_{n_k}\to \check{x}$ in the sense of Wigner measures,
see \cref{sec:wign-meas-husimi} below. We are now in a position to state our
second main result.

\begin{thm}[$\Gamma-$convergence of the relative Wehrl free energy]
  \label[theorem]{thm: Gamma}
  \mbox{}	\\
  On $ (X, \mathcal{T}) $,
  \begin{equation}
    \widetilde{F}_{\mathrm{W},\beta,\frac{1}{n}}  \xrightarrow[n \to \infty] {\Gamma} \widetilde{F}_{\mathrm{B},\beta}\;.
  \end{equation}
\end{thm}
\begin{corollary}[$\Gamma$-convergence of free energies]
  \label[corollary]{cor:1}
  \mbox{}	\\
  On $ (X, \mathcal{T}) $,
  \begin{equation}
    F_{\mathrm{W},\beta,\frac{1}{n}} - \tfrac{d\log (\pi/n)}{\beta} \xrightarrow[n \to \infty] {\Gamma} F_{\mathrm{B},\beta}\;.
  \end{equation}
\end{corollary}

\begin{remark}[Convergence of minima and minimizers]
  \mbox{}	\\
  $ \Gamma-$convergence between functionals implies convergence of the minima
  and of the corresponding minimizers (w.r.t.\ the topology $ \mathcal{T}
  $). This has important physical implications, since it follows that any
  quantum equilibrium state (solving the Wehrl Gibbs
  variational problem) converges, as $ \varepsilon \to 0 $, to the classical
  Gibbs state, solving the classical Gibbs variational problem. A second
  important property inherited from $ \Gamma-$convergence is stability of the
  minimizers under perturbations which are continuous w.r.t.\ the topology $
  \mathcal{T} $; on the other hand, this might not be so surprising, since
  such perturbations are actually quite regular.
\end{remark}
\begin{remark}[$\Gamma$--convergence for the von Neumann free energy]
  \label{rem:2}
  \mbox{}\\
  For the von Neumann free energy, the $\Gamma$--lower bound holds (see
  \cref{sec:von-neumann-free} below). However, it is not clear whether the
  $\Gamma$--upper bound shall hold as well: it is possible to construct sequences
  of quantum states that converge to a wide class of $\mu\in \mathcal{P}(\mathcal{H})$ in topology
  $\mathcal{T}$, for which the relative entropy diverges.\footnote{More precisely, for
    such sequences the von Neumann entropy converges after a
    renormalization/rescaling that diverges \emph{more slowly} than the Gibbs
    renormalization $d\log(\pi\varepsilon)$, see \cref{sec:stat-with-diff}.} Nonetheless,
  a $\Gamma$--upper bound could still be possible either by choosing a better trial
  state, or by restricting the class of states for which the functional is
  defined. However, the Heisenberg uncertainty principle is a major
  obstruction towards the construction of such trial states and may be
  ultimately responsible of the failure of the convergence.
\end{remark}

\bigskip

\subsubsection*{Acknowledgements}
\label{sec:acknowledgements}

$\phantom{,}\!\!$ {\footnotesize M.\ Correggi \& M.\ Falconi acknowledge the
  supports of PNRR Italia Domani and Next Generation EU through the ICSC
  National Research Centre for High Performance Computing, Big Data and
  Quantum Computing. M.\ Correggi, M.\ Falconi \& R.\ Gautier also
  acknowledge the MUR grant ``Dipartimento di Eccellenza 2023-2027'' of
  Dipartimento di Matematica, Politecnico di Milano, and the PRIN 2022 grant
  ``ONES -- OpeN and Effective quantum Systems'', prot.\ 2022L45WA3.}

\section{Semiclassical Analysis}
\label{sec:semicl-analys}

In this section, we review the aspects of semiclassical analysis we rely upon
the most.

\subsection{Upper and Lower Symbols}
\label{sec:wick-anti-wick}

Given an operator $A_{\varepsilon}$ on $\mathcal{H}_{\varepsilon}$, such that
the coherent states $\lvert z_{\varepsilon}\rangle\in D(A_{\varepsilon})$ for
all $z\in \mathcal{H}$, we can define its {\it lower symbol}
$a_{\varepsilon}^{\mathrm{low}}: \mathcal{H}\to \mathbb{C}$ as
\begin{equation}
  \label{eq:lower}
  a^{\mathrm{low}}_{\varepsilon}(z) : = \left\langle z_\varepsilon\left|A_\varepsilon\right|z_\varepsilon\right\rangle_{\mathcal{H}_{\varepsilon}}\; . 
\end{equation}
Whenever an operator is the Wick quantization of a symbol $a$, \emph{i.e.}\
$A_{\varepsilon}= a^{\mathrm{Wick}}$, then its lower symbol
$a^{\mathrm{low}}_{\varepsilon}$ exists, and furthermore,
$a^{\mathrm{low}}_{\varepsilon}=a$.

The {\it upper symbol} of an operator $A_\varepsilon$, on the other hand, is
defined implicitly as the function $
a^{\mathrm{up}}_{\varepsilon}:\mathcal{H}\to \mathbb{C} $, if it exists,
satisfying
\begin{equation}
  \label{eq:upper}
  A_\varepsilon = (a_{\varepsilon}^{\mathrm{up}})^{\mathrm{a-Wick}}=\int_{\mathcal{H}} a^{\mathrm{up}}_{\varepsilon}(z) \left|z_\varepsilon \right\rangle \left\langle z_\varepsilon\right| \frac{\mathrm{d} z}{(\pi \varepsilon)^d}.
\end{equation}
The relation between the upper and lower symbols is given -- at least for a
suitable class of symbols -- by the proposition below.
\begin{proposition}[Upper and lower symbols]
  \label[proposition]{pro:symbols}
  \mbox{}	\\
  Let $b \in \mathcal{P}_{p,q}$. The upper symbol
  ${b}^{\mathrm{up}}_\varepsilon$ of $B_{\varepsilon}=b^{\mathrm{Wick}}$
  exists and there exists $b_{k,\ell} \in \mathcal{P}_{k,\ell}$, for $k \leq
  p, \ell \leq q$ with $ k+\ell < p+q$, independent of $ \varepsilon $, such
  that
  \begin{equation}
    \label{eq:expansion}
    {b}^{\mathrm{up}}_\varepsilon (z) = b^{\mathrm{low}}_{\varepsilon}(z) + \sum_{\substack{k \leq p,  \ell \leq q \\ k+\ell < p+q }}  \varepsilon^{p+q-(k+\ell)} b_{k,\ell}(z)=b(z) + \sum_{\substack{k \leq p,  \ell \leq q \\ k+\ell < p+q }}  \varepsilon^{p+q-(k+\ell)} b_{k,\ell}(z)
  \end{equation}
  In particular, we have the following pointwise convergence of symbols:
  \begin{equation}
    \label{eq:pointwise}
    {b}^{\mathrm{up}}_\varepsilon (z) \xrightarrow[\varepsilon \to  0]{} b(z), \qquad	\forall z \in \mathcal{H}.
  \end{equation}
\end{proposition}
\begin{proof}
  Since $b\in \mathcal{P}_{p,q}$ and $\dim \mathcal{H} =d <\infty$, $b$ can
  be written as
$$
b(z) = \sum_{\left\lvert i\right\rvert = q, \left\lvert j\right\rvert = p }
\beta_{i,j} \overline{z}^i z^j,
$$
where $\beta_{i,j} \in \mathbb{C}$ and we used the multi-index notations: $i =
(i_1,...,i_d) \in \mathbb{N}^d$, $\left\lvert i\right\rvert := i_1 + ... +
i_d$ and $z^i := z_1^{i_1} ... z_d^{i_d}$. By linearity of the Wick
quantization \eqref{eq:Wick} and the upper symbol definition
\eqref{eq:upper}, it is sufficient to prove the results for the monomial
$c(z) = \overline{z}^i z^j = \overline{z}_1^{i_1} \cdots
\overline{z}_d^{i_d} {z}_1^{j_1} \cdots {z}_d^{j_d}$, where $ z_\ell =
\left\langle e_\ell|z \right\rangle_{\mathcal{H}} $, for some orthonormal
basis $ e_1, \ldots, e_d $. However, by \cite[Prop 2.7]{AmmariNier2008},
\begin{align*}
  c^{\mathrm{Wick}} &= \left(a_{\varepsilon}^*(e_1) \right)^{i_1} \cdots \left(a_{\varepsilon}^*(e_d) \right)^{i_d} a^{j_1}_\varepsilon(e_1) \cdots a^{j_d}_{\varepsilon}(e_d) =   (\overline{z}^i)^{\mathrm{Wick}} (z^j)^{\mathrm{Wick}} \\
  &=   (z^j)^{\mathrm{Wick}} (\overline{z}^i)^{\mathrm{Wick}} + \sum_{k=1}^{\min (p,q)} \frac{\varepsilon^k}{k!} \left(   \{ \overline{z}^i , z^j \}^{k}   \right)^{\mathrm{Wick}},
\end{align*}
where $\{ \cdot, \cdot \}$ is the Poisson bracket. Now, it suffices to note
that the upper symbol of $ (z^j)^{\mathrm{Wick}}
(\overline{z}^i)^{\mathrm{Wick}}$ is $z^j \overline{z}^i = c (z)$. By
induction on $p+q$, we get that the upper symbol is well defined and of the
wanted form.

\end{proof}

\subsection{Wigner Measures}
\label{sec:wign-meas-husimi}

Let us now introduce the key concept for our subsequent analysis:
\emph{Wigner measures}. Given a family of quantum states $ \left(
  \rho_{\varepsilon} \right)_{\varepsilon \in (0,1)} \subset
\mathfrak{S}_+^1(\mathcal{H}_{\varepsilon})$, a (Borel Radon) measure $\mu\in
\mathcal{M}(\mathcal{H})$ on the phase space $\mathcal{H}$ is a Wigner
measure for such family if and only if there exists a sequence
$\varepsilon_n\to 0$ such that for any symbol $b\in C_0(\mathcal{H})$,
\begin{equation*}
  \lim_{n\to \infty} \Tr_{\mathcal{H}_{\varepsilon_n}}\bigl(\rho_{\varepsilon_n} \, b^{\mathrm{a-Wick}}\bigr)= \int_{\mathcal{H}}^{}b(z)  \mathrm{d}\mu(z)\;.
\end{equation*}
We also say that $\rho_{\varepsilon_n}\to \mu$ \emph{in the sense of Wigner measures}. Let us
recall that convergence in the sense of Wigner measures is implied by
convergence \emph{in the sense of Fourier transforms}: if for any $z\in \mathcal{H}$,
\begin{equation*}
  \lim_{n\to \infty} \Tr_{\mathcal{H}_{\varepsilon_n}}\bigl(\rho_{\varepsilon_n} \, W_{\varepsilon_n}(\zeta)\bigr)= \int_{\mathcal{H}}^{}e^{2i\Re \langle \zeta  \vert z  \rangle_{\mathcal{H}}}  \mathrm{d}\mu(z)\;,
\end{equation*}
then $\rho_{\varepsilon_n}\to \mu$ in the sense of Wigner measures (see, \emph{e.g.},
\cite{AmmariNier2008}).

The Wigner measures of a given family of quantum states characterize such
states in the semiclassical limit, and this is typically the first step in
rigorously establishing the \emph{Bohr correspondence principle} for a
physical system. Let us remark that in general the Wigner measures can have
any mass between zero and one, and the loss happens whenever the mass can
escape at infinity in either position or momentum (or both) degrees of
freedom.

\begin{lemma}
  \label[lemma]{lemma:1}
  Let $ \left( \rho_{\varepsilon} \right)_{\varepsilon \in (0,1)} \subset \mathfrak{S}_+^1(\mathcal{H}_{\varepsilon})$ be a
  family of quantum states, and $\left(\varphi_{\varepsilon}\right)_{\varepsilon\in (0,1)}$ the
  corresponding Husimi probability measures ($\mathrm{d}\varphi_{\varepsilon}(z)=
  (\pi\varepsilon)^{-d}f_{\varepsilon}(z)\mathrm{d}z$).

  Then, the Lebesgue-a.e.\ pointwise convergence of
  $(\pi\varepsilon_n)^{-d}f_{\varepsilon_n}$ to $f_0$ implies that
  $\rho_{\varepsilon_n}$ converges to $\mathrm{d}\varphi_0(z)
  =f_0(z)\mathrm{d}z$ in the sense of Wigner measures.
\end{lemma}
\begin{proof}
  By definition of anti-Wick quantization, we can rewrite
  \begin{equation*}
    \Tr_{\mathcal{H}_{\varepsilon}}\bigl(\rho_{\varepsilon} \, b^{\mathrm{a-Wick}}\bigr)= \int_{\mathcal{H}}^{}b(z) \mathrm{d}\varphi_{\varepsilon}(z)= \int_{\mathcal{H}}^{}b(z) f_{\varepsilon}(z)\frac{\mathrm{d}z}{(\pi\varepsilon)^d}\;.
  \end{equation*}
  Since almost everywhere pointwise convergence of probability densities
  yields the weak convergence of measures, and $C_0(\mathcal{H})\subset
  C_{\mathrm{b}}(\mathcal{H})$, the result holds.
\end{proof}
Observe that the converse is not true: while pointwise convergence of
densities implies weak convergence of measures, the converse does not hold in
general.

Given a family of quantum states, the set of its Wigner measures is never
empty; however, all Wigner measures could have mass strictly less than
one. There is, interestingly, a very natural sufficient condition to ensure
that all Wigner measures of a given family are probabilities. This condition
is stated in the proposition below, in this form it is due to Ammari and Nier
\cite[Thm. 6.2]{AmmariNier2008}, and it holds for arbitrary Fock spaces
(\emph{i.e.}, also when $\mathcal{H}$ is infinite dimensional); to the best
of our knowledge, the result for a finite dimensional $\mathcal{H}$ first
appeared in \cite{LionsPaul}.

\begin{proposition}[No loss of mass]
  \label[proposition]{pro:Wigner}
  \mbox{}	\\
  Let $ \left( \rho_{\varepsilon} \right)_{\varepsilon \in (0,1)} \subset
  \mathfrak{S}_+^1(\mathcal{F}_\varepsilon(\mathcal{H})) $ be a family of normalized
  quantum states in the Fock space. Assume that, for some $ \delta > 0 $,
  \begin{equation}
    \Tr_{\mathcal{F}_\varepsilon(\mathcal{H})} \left( N_{\varepsilon}^{\delta} \rho_{\varepsilon} \right) \leq C < +\infty,
  \end{equation}
  uniformly in $ \varepsilon \in (0,1) $. Then, there exist a subsequence $ \left( \varepsilon_n
  \right)_{n \in \mathbb{N}} $, $ \varepsilon_n \to 0 $ as $ n \to + \infty $, and a probability measure $
  \mu \in \mathscr{P}(\mathcal{H}) $ that is a Wigner measure for the family, \emph{i.e.}
  $\forall b \in C_0(\mathcal{H}) $,
  \begin{equation}
    \lim_{n \to +\infty} \Tr_{\mathcal{F}_\varepsilon(\mathcal{H})} \left( b^{\mathrm{a-Wick}}_{\varepsilon_n} \rho_{\varepsilon_n} \right) = \int_{\mathcal{H}} b(z) \mathrm{d} \mu(z)\;.
  \end{equation}
  In addition,
  \begin{displaymath}
    \int_{\mathcal{H}} \left| z \right|^{2\delta} \mathrm{d} \mu(z) \leq C\;.
  \end{displaymath}
  Furthermore, any Wigner measure of the family $ \left( \rho_{\varepsilon} \right)_{\varepsilon \in
    (0,1)} $ is a probability, with bounded moments up to degree $2\delta$.
\end{proposition}

\section{Gibbs Entropy Convergence}
\label{sec:gibbs-entr-conv}

This section is devoted to the proof of \cref{thm:convergence}.

\subsection{Convergence of Partition Functions}
\label{sec:conv-part-funct}

The starting point is to prove convergence of partition functions. This is
done exploiting the convexity of the exponential function, and the relation
between an operator and its upper and lower symbols.

\begin{lemma}
  \label[lemma]{lemma:2}
  Let $Z_{\beta,\varepsilon}= \mathrm{Tr}_{\mathcal{F}_\varepsilon(\mathcal{H})} (e^{-\beta h^{\mathrm{Wick}}})$ be the quantum
  partition function at temperature $\beta^{-1}$, $Z_{\beta,0}= \int_{\mathcal{H}}^{}e^{-\beta h(z)}
  \mathrm{d}z$ the corresponding classical partition function. Then,
  \begin{equation*}
    \lim_{\varepsilon\to 0}(\pi\varepsilon)^d Z_{\beta,\varepsilon}= Z_{\beta,0}\;.  
  \end{equation*}
\end{lemma}
\begin{proof}

  Consider the upper symbol ${h}^{\mathrm{up}}_\varepsilon$ of $H_{\varepsilon}=
  h^{\mathrm{Wick}}$. By convexity of $x \mapsto e^{-\beta x}$, we have (see
  \cite[Chpt. 5.2.4]{Berezin1991}):
  \begin{equation*}
    \int_{\mathcal{H}} e^{-\beta h(z)} \mathrm{d} z \leq (\pi \varepsilon)^d \text{Tr}_{\mathcal{F}_\varepsilon(\mathcal{H})}\,e^{-\beta H_\varepsilon} \leq \int_{\mathcal{H}} e^{-\beta {h}^{\mathrm{up}}_\varepsilon(z)} \mathrm{d} z\;.
  \end{equation*}
In order to apply the dominated convergence theorem on the right hand side,
  we need to estimate from below the upper symbol: thanks to \eqref{eq:expansion}, we have
  \begin{equation*} {h}^{\mathrm{up}}_\varepsilon(z) =
  h(z)+ \sum_{\substack{k \leq p, \ell
        \leq q \\ k+\ell < p+q }} \varepsilon^{p+q-(k+\ell)} h_{k,\ell}(z)\;,
  \end{equation*}
	Since $h(z) \geq C \left\| z \right\|^{p+q}-\widetilde{C}$, the set $\{ h(z) \leq 1 \}$ is compact. There exists $R>0$ such that $\{ h(z) \leq 1 \} \subset \overline{B}(0,R)$. It is clear that ${h}^{\mathrm{up}}_\varepsilon(z) \mathds{1}_{ \overline{B}(0,R) }(z)$ is bounded uniformly w.r.t.\ $\varepsilon$. Since $h_{k,l}$ has degree strictly less than $p+q$, 
	$$
	\left\lvert\mathds{1}_{ \overline{B}(0,R)^c }(z) \frac{h_{k,l}(z)}{h(z)} \right\rvert \leq C \mathds{1}_{ \overline{B}(0,R)^c }(z) \frac{\langle\left\| z \right\| \rangle^{p+q-1}}{\left\| z \right\|^{p+q}} \leq C(R)
	$$
  is also bounded. This shows the bound on the remainder 
  $$\left\lvert \mathds{1}_{ \overline{B}(0,R)^c }(z) \sum_{\substack{k \leq p, \ell
        \leq q \\ k+\ell < p+q }} \varepsilon^{p+q-(k+\ell)} h_{k,\ell}(z)\ \right\rvert \leq \varepsilon C(R)  h(z).$$
        Hence, for $\varepsilon$ small enough (independent of $z$),
\begin{align*}
{h}^{\mathrm{up}}_\varepsilon(z) &\geq \mathds{1}_{ \overline{B}(0,R)^c }(z)  \frac{h(z)}{2} - C \\
& \geq \frac{h(z)}{2} - \widetilde{C}.
\end{align*}

\end{proof}

\subsection{Convergence of the Quantum Gibbs State in the Sense of Wigner
  Measures}
\label{sec:pointw-conv-husimi}

Let us now prove that the Husimi function of the quantum Gibbs state
converges pointwise to the classical Gibbs density, thus implying that the
classical Gibbs measure is the unique Wigner measure of the quantum Gibbs
state.

\begin{lemma}
  \label[lemma]{lemma:3}
  For any $z\in \mathcal{H}$, let $g_{\beta,\varepsilon}(z)= \left\langle z_\varepsilon\left|e^{-\beta
        H_\varepsilon}\right|z_\varepsilon\right\rangle$. Then,
  \begin{equation*}
    \lim_{\varepsilon\to 0} g_{\beta,\varepsilon}(z)= e^{-\beta h(z)}\;.
  \end{equation*}
\end{lemma}
\begin{corollary}
  \label[corollary]{cor:2}
  Let $f_{\beta,\varepsilon}(z)$ be the Husimi function of the quantum Gibbs state
  $\Gamma_{\beta,\varepsilon}$, and $\varphi_{\beta,\varepsilon}$ the corresponding probability measure. Then,
  \begin{gather*}
    \lim_{\varepsilon\to 0} (\pi\varepsilon)^{-d}f_{\beta,\varepsilon}= \tfrac{1}{Z_{\beta,0}}\,e^{-\beta h}\;,\\
    \lim_{\varepsilon\to 0}\varphi_{\beta,\varepsilon} = \gamma_{\beta}\;,
  \end{gather*}
  where the first limit is pointwise, the second is in the weak topology of
  measures, and $\gamma_{\beta}$ is the classical Gibbs measure.
\end{corollary}
\begin{proof}[Proof of \cref{lemma:3}]
  Let $z \in \mathcal{H}$ and $\varepsilon \in (0,1]$ both be fixed. Let us compute the derivative
  with respect to $\beta$ of $g_{\varepsilon,z}(\beta) = g_{\beta,\varepsilon}(z)$:
  \begin{align*}
    - \frac{\mathrm{d} g_{\varepsilon,z}}{\mathrm{d} \beta}
    &=  \left\langle z_\varepsilon\left|H_\varepsilon e^{-\beta H_\varepsilon}\right|z_\varepsilon\right\rangle \\
    &=  \left\langle \Omega\left|W \left( \tfrac{\sqrt{2}z }{i \varepsilon} \right)^* h^{\mathrm{Wick}} W \left( \tfrac{\sqrt{2}z }{i \varepsilon} \right) W \left( \tfrac{\sqrt{2}z }{i \varepsilon} \right)^* e^{- \beta H_\varepsilon} W \left( \tfrac{\sqrt{2}z }{i \varepsilon} \right)\right|\Omega\right\rangle \\
    &= \left\langle \Omega\left|(h(\cdot+z))^{\mathrm{Wick}}  W \left( \tfrac{\sqrt{2}z }{i \varepsilon} \right)^* e^{- \beta H_\varepsilon} W \left( \tfrac{\sqrt{2}z }{i \varepsilon} \right)\right|\Omega\right\rangle,
  \end{align*}
  where we used the identity
  \begin{equation*}
    W \left( \tfrac{\sqrt{2}z }{i \varepsilon} \right)^* h^{\text{Wick}} W \left( \tfrac{\sqrt{2}z }{i \varepsilon} \right)  = (h(\cdot+z))^{\text{Wick}} \;.
  \end{equation*}
  Now, using the fact that $h$ is a sum of homogeneous polynomials, we obtain
  that
  \begin{equation}
    \label{eqp:derivativeb}
    -\frac{\mathrm{d} g_{\varepsilon,z}(\beta)}{\mathrm{d} \beta} = h(z) g_{\varepsilon,z}(\beta) + P(z,\beta,\varepsilon)\;, 
  \end{equation}
  with
  \begin{equation}
    \label{eqp:bound poly}
    \left\lvert P(z,\beta,\varepsilon)\right\rvert \leq \varepsilon^{\frac{1}{2}} Q(z)\;.
  \end{equation}
  Here, $Q(z)$ is a polynomial only in $z$: in order to see this, we expand
  the terms $ (a_{\varepsilon}(e_i) + z)^{j_i} $ appearing in $ (h(\cdot+z))^{\text{Wick}} $, to
  get expressions of the form
  \begin{displaymath}
    Q(z) \left\langle \Omega\left|a_{\varepsilon}^{j_1}(e_1) \cdots  a_{\varepsilon}^{j_d}(e_d) W \left( \tfrac{\sqrt{2}z }{i \varepsilon} \right)^* e^{- \beta H_\varepsilon} W \left( \tfrac{\sqrt{2}z }{i \varepsilon} \right)\right|\Omega\right\rangle,
  \end{displaymath}
  where $ Q(z) $ is a suitable polynomial in $ z $. The term with all $ j_i$s
  equal to zero yields the first term of the right hand side of
  \eqref{eqp:derivativeb}. For all the other terms, one can use
  Cauchy-Schwarz inequality to estimate the absolute valute of the above
  expression:
  \begin{displaymath}
    |Q(z)| \left\| a_{\varepsilon}^{j_1}(e_1) \cdots  a_{\varepsilon}^{j_d}(e_d) \left|\Omega \right\rangle \right\| \left\| W \left( \tfrac{\sqrt{2}z }{i \varepsilon} \right)^* e^{- \beta H_\varepsilon} W \left( \tfrac{\sqrt{2}z }{i \varepsilon} \right)\left|\Omega \right\rangle \right\| \leq  \varepsilon^{\frac{\left\lvert j\right\rvert}{2}} |Q(z)| \;.
  \end{displaymath}
  This gives \eqref{eqp:bound poly}, since $ |j| \geq 1 $ for all the remaining
  terms. Now, combining \eqref{eqp:derivativeb} with Duhamel's formula, we
  get
  \begin{equation*}
    g_{\varepsilon,z} (\beta) = e^{-\beta h(z)} \left( 1- \int_0^\beta e^{\widetilde{\beta} h(z)} P(z,\widetilde{\beta},\varepsilon) \mathrm{d} \widetilde{\beta} \right).
  \end{equation*}
  Using the bound \eqref{eqp:bound poly} on $P$, we conclude the proof by
  dominated convergence.
\end{proof}

\subsection{Convergence of the Gibbs Wehrl Entropy}
\label{sec:upper-bound-entr}

In this section, we prove the following proposition.
\begin{proposition}[Convergence of Wehrl entropy]
  \label[proposition]{prop:3}
  \mbox{}\\
  The Gibbs renormalized Wehrl entropy converges to its classical
  counterpart:
  \begin{equation*}
    \lim_{\varepsilon\to 0}\Bigl(S_{\mathrm{W},\varepsilon}(\Gamma_{\beta,\varepsilon}) - \log Z_{\beta,\varepsilon}\Bigr) = S_{\mathrm{B}}(\gamma_{\beta}) - \log Z_{\beta,0}\;. 
  \end{equation*}
  Furthermore, we have the following ``von Neumann upper bound'':
  \begin{equation*}
    \limsup_{\varepsilon\to 0} \Bigl(S_{\mathrm{vN},\varepsilon}(\Gamma_{\beta,\varepsilon}) - \log Z_{\beta,\varepsilon}\Bigr) \leq S_{\mathrm{B}}(\gamma_{\beta}) - \log Z_{\beta,0}\; .
  \end{equation*}
\end{proposition}

By \cref{lemma:3}, we can reformulate $\eqref{cvg entropy}$ as
\begin{equation}
  \label{eqp:entropy convergence}
  S_{\diamond,\varepsilon} (\Gamma_{\beta,\varepsilon}) -\log Z_{\beta, \varepsilon} \xrightarrow[\varepsilon \rightarrow 0]{} S_{\mathrm{B}} ( \gamma_{\beta} ) - \log Z_{\beta, 0}\; .
\end{equation}
Furthermore, we have the following general strict inequality for $ \varepsilon > 0 $
(see \cite{MR0496300}).
\begin{proposition}[Upper bound on von Neumann entropy]
  \label[proposition]{prop:2}
  \mbox{}\\
  Let $\rho\in \mathfrak{S}^1_+ (\mathcal{H})$ be a density matrix. Then,
  \begin{equation*}
    S_{\mathrm{vN},\varepsilon} (\rho) < S_{\mathrm{W},\varepsilon} (\rho)\; .
  \end{equation*}
\end{proposition}
Hence, for an upper bound to the l.h.s.\ of \eqref{eqp:entropy convergence}
in the von Neumann case, it suffices to prove an upper bound (actually, the
convergence) in the Wehrl case.  We can rewrite the renormalized Wehrl
entropy in terms of the Husimi function $f_{\beta,\varepsilon}(z) = \langle z_\varepsilon | \Gamma_{\beta,\varepsilon} | z_\varepsilon \rangle$ as
\begin{align*}
  S_{\mathrm{W},\varepsilon} (\Gamma_{\beta,\varepsilon}) - \log Z_{\beta, \varepsilon} &= - \int f_{\beta,\varepsilon}(z) \log f_{\beta,\varepsilon}(z) \frac{\mathrm{d} z}{(\pi \varepsilon)^d} - \log Z_{\beta, \varepsilon} \\
  & = - \int \frac{f_{\beta,\varepsilon} (z)}{(\pi \varepsilon)^d} \log \left(  \frac{f_{\beta,\varepsilon} (z)}{(\pi \varepsilon)^d}  \right) \mathrm{d} z - \log \left(  (\pi \varepsilon)^d Z_{\beta, \varepsilon} \right).
\end{align*}
The second term on the r.h.s.\ converges to $-\log Z_{\beta, 0}$ by
\cref{lemma:2}. Therefore, it only remains to deal with the integral. This
is done by dominated convergence, exploiting \cref{lemma:3}.

To use dominated convergence however, it is necessary to have a uniform
$L^1(\mathcal{H})$-bound of the function $(\pi \varepsilon)^{-d} f_{\beta,\varepsilon}(z) \log \left( (\pi \varepsilon)^{-d}
  f_{\beta,\varepsilon}(z) \right)$. By convexity and Jensen inequality,
$$\frac{f_{\beta,\varepsilon}(z)}{(\pi \varepsilon)^d} = \frac{1}{(\pi \varepsilon)^d Z_{\beta, \varepsilon}}   \left\langle z_\varepsilon\left|e^{-\beta
      H_\varepsilon }\right| z_\varepsilon\right\rangle \geq \frac{1}{(\pi \varepsilon)^d Z_{\beta, \varepsilon}} e^{-\beta \left\langle z_\varepsilon\left|H_\varepsilon\right|z_\varepsilon\right\rangle} =\frac{1}{(\pi \varepsilon)^d Z_{\beta, \varepsilon}} e^{-\beta h(z)}\; .  $$
Hence,
\begin{equation}
  \label{eqp:husimiub}
  - \log \frac{f_{\beta,\varepsilon}(z)}{(\pi \varepsilon)^d} \leq \beta h(z) + \log \left(   (\pi \varepsilon)^d Z_{\beta, \varepsilon} \right) \;,
\end{equation}
and the standard inequality $ - x \log x \geq - x^2 $ implies that
\begin{equation}
  \label{eqp:husimilb}
  -\frac{f_{\beta,\varepsilon}(z)}{(\pi \varepsilon)^d} \log \frac{f_{\beta,\varepsilon}(z)}{(\pi \varepsilon)^d} \geq - \left( \frac{f_{\beta,\varepsilon}(z)}{(\pi \varepsilon)^d} \right)^2\;.
\end{equation}
Notice that $h(z)$ is independent of $\varepsilon$ and $\log \left( (\pi \varepsilon)^d Z_{\beta, \varepsilon}
\right) = \mathcal{O} (1)$ is independent of $z$, therefore it suffices to find a
uniform bound on $(\pi \varepsilon)^{-d} f_{\beta,\varepsilon}(z)$ decaying fast enough to ensure
integrability. In particular, we may restrict to consider the quantity
$$
g_{\beta,\varepsilon}(z) = \left\langle z_\varepsilon\left|e^{-\beta H_\varepsilon}\right|z_\varepsilon\right\rangle\;,
$$
defined above in \cref{sec:pointw-conv-husimi}. Here however, we focus on the
decay in $ z $ of such function. Since $\left\lvert g_{\beta,\varepsilon}(z)\right\rvert \leq 1$ uniformly
in $ \varepsilon $ and $ z $, it suffices to prove a decay for large $ |z| $. Hence, we
can assume that $\left\lvert z_i\right\rvert \geq 1$, for every $i \in \{1,\ldots,d\}$. Here, we
exploit assumption \eqref{A}, to estimate for any $ k \in \mathbb{N} $
\begin{align*}
  g_{\beta,\varepsilon} (z) &= \left\langle z_\varepsilon\left|(N_\varepsilon +\varepsilon)^{-\frac{k}{2}} (N_\varepsilon +\varepsilon)^{\frac{k}{2}} e^{-\beta H_\varepsilon} (N_\varepsilon +\varepsilon)^{\frac{k}{2}} (N_\varepsilon +\varepsilon)^{-\frac{k}{2}}\right|z_\varepsilon\right\rangle \\
  & \leq \left\| (N_\varepsilon +\varepsilon)^{\frac{k}{2}} e^{-\beta H_\varepsilon} (N_\varepsilon +\varepsilon)^{\frac{k}{2}} \right\|  \left\langle z_\varepsilon\left|(N_\varepsilon+\varepsilon)^{-k}\right|z_\varepsilon\right\rangle \\
  &\leq C_k \left\langle z_\varepsilon\left|(N_\varepsilon+\varepsilon)^{-k}\right|z_\varepsilon\right\rangle\; .
\end{align*}
The last factor can be computed explicitly:
\begin{align*}
  \left\langle z_\varepsilon\left|(N_\varepsilon+1)^{-k}\right|z_\varepsilon\right\rangle & = e^{ - \frac{ \left\lvert z\right\rvert^2}{\varepsilon}} \sum_{n_1, \ldots, n_d \in \mathbb{N}} \left(  \tfrac{\left\lvert z_1\right\rvert^2}{\varepsilon}  \right)^{n_1} \cdots \left(  \tfrac{\left\lvert z_d\right\rvert^2}{\varepsilon}  \right)^{n_d} \prod_{j =1}^d  \frac{1}{n_j !} \frac{1}{(\varepsilon n_j+\varepsilon)^{k}}	\\
  & = e^{ - \frac{ \left\lvert z\right\rvert^2}{\varepsilon}} \prod_{j =1}^d \sum_{n_j \in \mathbb{N}}  \left( \tfrac{\left\lvert z_j\right\rvert^2}{\varepsilon}  \right)^{n_j}  \frac{1}{n_j !}  \frac{1}{(\varepsilon n_j+\varepsilon)^{k}}\;.
\end{align*}
Now, for any $ k \geq 1 $,
\begin{align*}
  \frac{k!n!(n+1)^k}{(n+k)!} & = \frac{k! (n+1)^k}{(n+k)(n+k-1) \cdots (n+1)} = \frac{k!}{1 \cdot \left( 1 + \frac{1}{n+1} \right) \cdots \left( 1 + \frac{k-1}{n+1} \right)} \\
  & \geq  \frac{k!}{k!} = 1\;,
\end{align*} 
since $ 1 + \frac{j-1}{n+1} \leq j $, for any $ j \geq 1 $. Hence, we get
$$
\frac{1}{(\varepsilon n_j +\varepsilon)^{k}} \leq \frac{ k! n_j! }{ \varepsilon^{k} (n_j+k)! }\; ,
$$
which implies
\begin{align*}
  \left\langle z_\varepsilon\left|(N_\varepsilon+1)^{-k}\right|z_\varepsilon\right\rangle  & \leq \prod_{j =1}^d e^{ - \frac{ \left\lvert z_j\right\rvert^2}{\varepsilon}}  \sum_{n_j \in \mathbb{N}}  \left( \tfrac{\left\lvert z_j\right\rvert^2}{\varepsilon}  \right)^{n_j}  \frac{ k! }{ \varepsilon^{k} (n_j+k)! } \\
  & \leq \prod_{j =1}^d \frac{k!}{\varepsilon^{k}} \left( \tfrac{\varepsilon}{\left\lvert z_j\right\rvert^2} \right)^{k} e^{ - \frac{ \left\lvert z_j\right\rvert^2}{\varepsilon}}  \sum_{n_j \in \mathbb{N}, n_j \geq k}  \frac{1}{n_j! } \left( \tfrac{\left\lvert z_j\right\rvert^2}{\varepsilon}  \right)^{n_j}   \leq \prod_{j=1}^d \frac{k!}{\left\lvert z_j\right\rvert^{2k}}\; .
\end{align*}
Since, thanks to \eqref{A}, the bound holds for any $k \in \mathbb{N}$, the above estimate shows
that, by taking $ k $ large enough, we can ensure $L^1$-integrability of $
g_{\beta,\varepsilon}(z) $ uniformly in $ \varepsilon $ and $ \beta $. Furthermore, $ Z_{\beta, \varepsilon} $ converges to
a bounded constant, so that, combining \eqref{eqp:husimiub} with
\eqref{eqp:husimilb}, we get, for $ |z_j| \geq 1 $,
\begin{align}
  \left|  \frac{f_{\beta,\varepsilon}(z)}{(\pi \varepsilon)^{d}} \log \left(  \frac{f_{\beta,\varepsilon}(z)}{(\pi \varepsilon)^{d}}\right) \right| & \leq \left( \beta h(z) + \log \left(   (\pi \varepsilon)^d Z_{\beta, \varepsilon} \right) \right) \frac{f_{\beta,\varepsilon}(z)}{(\pi \varepsilon)^{d}} + \left( \frac{f_{\beta,\varepsilon}(z)}{(\pi \varepsilon)^{d}} \right)^2 \nonumber \\
  & \leq C'_k  \left(1 + h(z) \right) \prod_{j=1}^d \frac{1}{\left\lvert z_j\right\rvert^{2k}}\; , \label{eqp:Husimi bound}
\end{align}
which belongs to $ L^1(\mathcal{H}) $, if $ k $ is large enough, since $ h $ is a
polynomial in $ z $.

\subsection{Lower Bound to the von Neumann entropy}
\label{sec:lower-bound-von}

Let $ \left\{ e_i \right\}_{i \in \mathbb{N}}  $ be an orthonormal basis of $\mathcal{F}_\varepsilon(\mathcal{H})$; then, using the upper symbol representation of $ H_{\varepsilon} $, one gets
\begin{align*}
	S_{\mathrm{vN},\varepsilon} \left( \Gamma_{\beta,\varepsilon} \right) - \log Z_{\beta, \varepsilon} &= \frac{\beta}{Z_{\beta, \varepsilon}} \text{Tr}_{\mathcal{H}_\varepsilon} \left(  H_\varepsilon e^{-\beta H_\varepsilon}  \right) \\
&=  \frac{\beta}{Z_{\beta, \varepsilon}}\text{Tr}_{\mathcal{H}_\varepsilon} \left(  \int_{\mathcal{H}}^{} h_{\varepsilon}^{\mathrm{up}}(z)  e^{-\frac{\beta}{2} H_{\varepsilon}} \lvert z_{\varepsilon}\rangle\langle z_{\varepsilon}\rvert e^{-\frac{\beta}{2} H_\varepsilon}\tfrac{\mathrm{d}z}{(\pi\varepsilon)^{d}}  \right)   \\
& =  \frac{\beta}{Z_{\beta, \varepsilon}}\lim_{N\to \infty} \int_{\mathcal{H}} \mathrm{d} z \: {h}^{\mathrm{up}}_\varepsilon (z) \left\langle z_{\varepsilon}\left|e^{-\frac{\beta}{2} H_{\varepsilon}} \sum_{j=1}^N\lvert e_j\rangle\langle e_j\rvert e^{- \frac{\beta}{2}H_{\varepsilon}}\right|z_{\varepsilon}\right\rangle\\
&= \frac{\beta}{Z_{\beta, \varepsilon}} \int_{\mathcal{H}} \mathrm{d} z \: {h}^{\mathrm{up}}_\varepsilon (z) \left\langle z_{\varepsilon}\left|e^{-\beta H_{\varepsilon}}\right|z_{\varepsilon}\right\rangle\;,
\end{align*}
where in the last equality the limit can go inside the integral by monotone
convergence, since the integrand is a sum with an increasing number of
non-negative terms. Now, exploiting the convexity of $ x \mapsto e^{-\beta x} $ and
applying Jensen inequality, we get
\begin{align*}
	S_{\mathrm{vN},\varepsilon} \left( \Gamma_{\beta,\varepsilon} \right) - \log Z_{\beta, \varepsilon} & \geq \frac{\beta}{Z_{\beta, \varepsilon}} \int_{\mathcal{H}} \mathrm{d} z \: {h}^{\mathrm{up}}_\varepsilon (z) e^{-\beta \left\langle z_{\varepsilon}\left|H_{\varepsilon}\right|z_{\varepsilon}\right\rangle}	\\
	& = \frac{\beta}{Z_{\beta, \varepsilon}} \int_{\mathcal{H}} \mathrm{d} z \: {h}^{\mathrm{up}}_\varepsilon (z) e^{-\beta {h}(z)}\;,
\end{align*}
so the result follows by dominated convergence and \cref{pro:symbols}.

\section{$\Gamma$-Convergence of Free Energies}
\label{sec:gamma-conv-free-1}

In this section, we prove \cref{thm: Gamma}.

\subsection{$\Gamma$-Lower Bound}
\label{sec:gamma-lower-bound}

Let us recall that the $\Gamma$-lower bound can be formulated as follows:
\begin{itemize}
\item $\forall x\in X$, $\forall  (x_n)_{n\in \mathbb{N}}\subset X$ such that $x_n \overset{\mathcal{T}}{\to} x$:
  \begin{equation*}
    \widetilde{F}_{\mathrm{B},\beta}(x) \leq \liminf_{n\to \infty} \widetilde{F}_{\diamond,\beta,\frac{1}{n}}(x_n)\;.
  \end{equation*}
\end{itemize}

We make crucial use of the following result.
\begin{proposition}[\cite{Lewin2015}, Thm.\ 7.1]
  \label[proposition]{prop:4}
  \mbox{}\\
  Let $\rho_{\varepsilon},\sigma_{\varepsilon}\in \mathfrak{S}^1_{+}(\mathcal{F}_\varepsilon(\mathcal{H})) $. Then,
  \begin{equation*}
    S_{\mathrm{vN},\varepsilon}(\rho_{\varepsilon}\Vert\sigma_{\varepsilon})\geq S_{\mathrm{W},\varepsilon}(\rho_{\varepsilon}\Vert\sigma_{\varepsilon})\;.
  \end{equation*}
  Furthermore, if $\rho_{\varepsilon_n}\to \mu$ and $\sigma_{\varepsilon_n}\to \nu$ in the sense of Wigner
  measures, then
  \begin{equation*}
    \liminf_{n\to \infty} S_{\mathrm{vN},\varepsilon_n}(\rho_{\varepsilon_n}\Vert\sigma_{\varepsilon_n})\geq \liminf_{n\to \infty} S_{\mathrm{W},\varepsilon_n}(\rho_{\varepsilon_n}\Vert \sigma_{\varepsilon_n})\geq S_{\mathrm{B}}(\mu\Vert\nu)\;.
  \end{equation*}
\end{proposition}

\begin{remark}[Independence of the relative entropy on the Husimi
  semiclassical parameter]
  \label[remark]{rem:1}
  \mbox{}\\
  By our definition, $S_{\mathrm{W},\varepsilon}(\rho\Vert\sigma)= S_{\mathrm{B}}(\varphi_{\varepsilon}\Vert \phi_{\varepsilon}) =
  \int_{\mathcal{H}}^{} f_{\varepsilon}(z) (\log (f_{\varepsilon}(z))-\log (g_{\varepsilon}(z)))\frac{\mathrm{d}z}{(\pi\varepsilon)^d}$, where
  $\varphi_{\varepsilon},\phi_{\varepsilon}$ and $f_{\varepsilon},g_{\varepsilon}$ are the Husimi probability distributions
  and functions of $\rho,\sigma$ respectively. However, the explicit dependence on
  the semiclassical parameter is deceiving: by rescaling, one can easily
  check that
  \begin{equation*}
    S_{\mathrm{B}}(\varphi_{\varepsilon}\Vert \phi_{\varepsilon})= S_{\mathrm{B}}(\varphi_1\Vert \phi_1)\;,
  \end{equation*}
  where $\mathrm{d}\varphi_1(z)= \langle z_1 \vert \rho \vert z_1 \rangle_{}\frac{\mathrm{d}z}{\pi^d}$ and
  $\mathrm{d}\phi_1(z)= \langle z_1 \vert \sigma \vert z_1 \rangle_{}\frac{\mathrm{d}z}{\pi^d}$.
\end{remark}

\subsubsection{von Neumann Free Energy}
\label{sec:von-neumann-free}

Let us start by the von Neumann free energy, for which the $\Gamma$-lower bound
follows almost directly from \cref{prop:4} above. In light of the definition,
the functionals $\widetilde{F}_{\diamond,\beta,\frac{1}{n}}(x)=+\infty$ whenever $x\notin
\mathfrak{G}_{\frac{1}{n}}$, and thus the only sequences $x_n\to x$ for which the
$\Gamma$--lower bound is not automatically satisfied are those such that there
exists a subsequence $n_k\to \infty$ such that $x_{n_k}\in \mathfrak{G}_{\frac{1}{n_k}}$ for any
$k\in \mathbb{N}$, and one can restrict along such subsequences in computing the limes
inferior. As we discussed in \cref{sec:gamma-conv-free}, the existence of
such a subsequence implies that $x\in \mathfrak{G}_0$; furthermore, the uniform
convergence on compact subsets of the Fourier transform implies convergence
of the respective Husimi measures $\mathrm{d}\xi_{n_k}(z)= \langle z_{\frac{1}{n_k}}
\vert \check{x}_{n_k}\vert z_{\frac{1}{n_k}} \rangle_{\mathcal{F}_\varepsilon(\mathcal{H})}\frac{\mathrm{d}z}{(\pi/n_k)^d}$
when tested on $\mathcal{C}_0(\mathcal{H})$, that by density extends to weak convergence since
the $\xi_{n_k}$ are all probabilities.

Furthermore, the Boltzmann relative entropy is
jointly lower semicontinuous: if $\mu_n \rightharpoonup \mu $ and $\nu_n \rightharpoonup \nu$ (weak convergence
of measures), then
\begin{equation*}
  \liminf_{n\to \infty} S_{\mathrm{B}}(\mu_n,\nu_n)\geq S_{\mathrm{B}}(\mu,\nu)\;.
\end{equation*}
Therefore, by joint lower semicontinuity and \cref{prop:4} we have that
\begin{equation*}
  \liminf_{n\to \infty} S_{\mathrm{vN}, \frac{1}{n}}(\check{x}_n\Vert \Gamma_{\beta,\frac{1}{n}})\geq \liminf_{n\to \infty} S_{\mathrm{B}} (\xi_n\Vert \varphi_{\beta, \frac{1}{n}})\geq S_{\mathrm{B}}(\check{x}\Vert \gamma_{\beta})\;;
\end{equation*}
where $\varphi_{\beta,n}$ is the Husimi probability of $\Gamma_{\beta, \frac{1}{n}}$:
$\mathrm{d}\varphi_{\beta,\frac{1}{n}}(z)= \langle z_{\frac{1}{n}} \vert \Gamma_{\beta, \frac{1}{n}}\vert z_{\frac{1}{n}}
\rangle_{\mathcal{F}_\varepsilon(\mathcal{H})}\frac{\mathrm{d}z}{(\pi/n)^d}$ (that converges weakly to $\gamma_{\beta}$, as
shown in \cref{cor:2}).

\subsubsection{Wehrl Free Energy}
\label{sec:wehrl-free-energy}

The argument for the Wehrl free energy goes along the exact same lines,
recalling that
\begin{equation*}
  \widetilde{F}_{\mathrm{W},\beta,\frac{1}{n}}(x_n)= \tfrac{1}{\beta}S_{\mathrm{B}}(\xi_n\Vert \gamma_{\beta,\frac{1}{n}})\;,
\end{equation*}
where $\xi_n$ is the Husimi probability of $\check{x}_n$ (with respect to the
semiclassical parameter $\frac{1}{n}$), and $\mathrm{d}\gamma_{\beta, \frac{1}{n}}=
\frac{e^{-\beta h_{1/n}^{\mathrm{up}}(z)}\mathrm{d}z}{\int_{\mathcal{H}}^{}e^{-\beta
    h^{\mathrm{up}}_{1/n}(z)} \mathrm{d}z}$. We have to make use of the lemma
below, in conjunction with the joint lower semicontinuity of the Boltzmann
relative entropy.

\begin{lemma}
  \label{lemma:4}
  $\gamma_{\beta, \frac{1}{n}} \rightharpoonup \gamma_{\beta}$.
\end{lemma}
\begin{proof}
  By \cref{pro:symbols}, $h^{\mathrm{up}}_{\varepsilon} \to h$ pointwise as $\varepsilon\to 0$. It is
  then enough to apply a dominated convergence argument on the numerator and
  denominator of
  \begin{equation*}
    \frac{\int_{\mathcal{H}}^{}b(z)  e^{-\beta h_{1/n}^{\mathrm{up}}(z)}\mathrm{d}z}{\int_{\mathcal{H}}^{}e^{-\beta h^{\mathrm{up}}_{1/n}(z)} \mathrm{d}z}\;,
  \end{equation*}
  for any $b\in C_{\mathrm{b}}(\mathcal{H})$, analogous to the one performed in the proof
  of \cref{lemma:2}.
\end{proof}

\subsection{$\Gamma$-Upper bound}
\label{sec:gamma-upper-bound}

The $\Gamma$-upper bound for the Wehrl relative free energy reads as follows:
\begin{itemize}
\item $\forall x\in X$, $\exists (x_n)_{n\in \mathbb{N}}\subset X$ such that $x_n\overset{\mathcal{T}}{\to }x$, and
  \begin{equation*}
    \limsup_{n\to \infty} \widetilde{F}_{\mathrm{W},\beta, \frac{1}{n}}(x_n) \leq \widetilde{F}_{\mathrm{B},\beta}(x)\;.
  \end{equation*}
\end{itemize}

\subsubsection{Wehrl Free Energy}
\label{sec:wehrl-free-energy-1}

Firstly, observe that whenever $x\in \mathfrak{G}_{\varepsilon}$, $\varepsilon>0$, or $\check{x}\in \mathcal{P}(\mathcal{H})$ and is
not absolutely continuous with respect to the Lebesgue measure, the upper
bound is trivially satisfied. Thus, we need to prove it holds for any
absolutely continuous measure $\check{x}= f(z)\mathrm{d}z$, $f\in L^1(\mathcal{H})$, $f\geq
0$, $\int_{\mathcal{H}}^{}f(z) \mathrm{d}z=1$.

Let us define the trial states
\begin{equation}
	\rho_{\varepsilon_n}(x) = \int_{\mathcal{H}} \left|z_{\varepsilon_n} \right\rangle\left\langle z_{\varepsilon_n}\right| f(z)\mathrm{d} z\;.
\end{equation}
It is well-known, see \emph{e.g.} \cite{AmmariNier2008}, that
\begin{displaymath}
	x_{\varepsilon_n}=\widehat{\rho_{\varepsilon_n}(x)} \xrightarrow[n \rightarrow +\infty]{\mathcal{T}} x\;,
\end{displaymath}
where $\mathfrak{G}_{\varepsilon_n}\ni x_{\varepsilon_n}=\widehat{\rho_{\varepsilon_n}(x)}$ is the noncommutative Fourier
transform of $\rho_{\varepsilon_n}(x)$, and
$$
\text{Tr}_{\mathcal{F}_\varepsilon(\mathcal{H})} \left(  H_{\varepsilon_n} \rho_{\varepsilon_n}(x) \right) = \text{Tr}_{\mathcal{F}_\varepsilon(\mathcal{H})} \left ( \rho_{\varepsilon_n}(x) h^\text{Wick} \right) \xrightarrow[n \to + \infty]{} \int_{\mathcal{H}} h(z) f(z)\mathrm{d} z\;.
$$

In view of the above, since the renormalized Wehrl free energy can be written as
\begin{equation*}
  F_{\mathrm{W},\beta,\varepsilon_n}\bigl(x_{\varepsilon_n}\bigr) - \tfrac{d}{\beta}\log(\pi\varepsilon_n)= \Tr_{\mathcal{H}_{\varepsilon_n}}\bigl(H_{\varepsilon_n}\rho_{\varepsilon_n}(x)\bigr) - \tfrac{1}{\beta} S_{\mathrm{W},\varepsilon_n}\bigl(\rho_{\varepsilon_n}(x)\bigr) - \tfrac{d}{\beta}\log(\pi\varepsilon_n)\;,
\end{equation*}
it remains to prove that
\begin{equation*}
  \limsup_{n\to \infty} - \tfrac{1}{\beta} S_{\mathrm{W},\varepsilon_n}\bigl(\rho_{\varepsilon_n}(x)\bigr) - \tfrac{d}{\beta}\log(\pi\varepsilon_n)\leq - \tfrac{1}{\beta}S_{\mathrm{B}}(x)\;,
\end{equation*}
or equivalently,
\begin{equation*}
  S_{\mathrm{B}}(x)\leq \liminf_{n\to \infty} S_{\mathrm{W},\varepsilon_n}\bigl(\rho_{\varepsilon_n}(x)\bigr)+ d\log(\pi\varepsilon_n)\;.
\end{equation*}
In fact, by the results of \cref{sec:gibbs-entr-conv},
$F_{\mathrm{B},\beta}(\gamma_{\beta,\varepsilon})\underset{\varepsilon\to 0}{\longrightarrow}
F_{\mathrm{B},\beta}(\gamma_{\beta})$, thus completing the proof for the relative entropy
as well, see \eqref{eq:8}.

The Husimi function $f_{\varepsilon_n}(z)$ of the trial state $\rho_{\varepsilon_n}(x)$ can be written as
\begin{multline*}
	f_{\varepsilon_n} (z) = \left\langle z_{\varepsilon_n}\left|\rho_{\varepsilon_n}(x)\right|z_{\varepsilon_n}\right\rangle = \int_{\mathcal{H}} \left\lvert\left\langle z_{\varepsilon_n}|w_{\varepsilon_n} \right\rangle\right\rvert^2 f(w) \mathrm{d} w \\= \int e^{-\frac{\left\lvert z-w\right\rvert^2}{2 {\varepsilon_n}}} f(w) \mathrm{d} w = \left( f \ast {e^{-\frac{\left\lvert\cdot\right\rvert^2}{{2\varepsilon_n}}}} \right) (z)\;.
\end{multline*}
From this it follows that
\begin{displaymath}
 	S_{\mathrm{W},\varepsilon_n} \bigl(\rho_{\varepsilon_n}(x)\bigr)  + d\log (\pi {\varepsilon_n}) = S_{\mathrm{B}} \left(f \ast \tfrac{e^{-\frac{\left\lvert\cdot\right\rvert^2}{{2\varepsilon_n}}}}{ (\pi {\varepsilon_n})^{d}}\right)\;.
\end{displaymath}
Now, by concavity of the function $ \xi: x \mapsto -x \log x$, we can apply Jensen inequality, so obtaining
\begin{displaymath}
\xi \left( f \ast \tfrac{e^{-\frac{\left\lvert\cdot\right\rvert^2}{{2\varepsilon_n}}}}{ (\pi {\varepsilon_n})^{d}} \right) (z) \geq \int \xi(f(z-y)) \tfrac{e^{-\frac{\left\lvert\cdot\right\rvert^2}{{2\varepsilon_n}}}}{ (\pi {\varepsilon_n})^{d}} \mathrm{d} y = \left( \xi(f) \ast  \tfrac{e^{-\frac{\left\lvert\cdot\right\rvert^2}{{2\varepsilon_n}}}}{(\pi {\varepsilon_n})^d} \right) (z)\;.
\end{displaymath}
Since $\tfrac{e^{-\frac{\left\lvert\cdot\right\rvert^2}{{2\varepsilon_n}}}}{(\pi {\varepsilon_n})^d}$ is a
rapidly decreasing mollifier in $\mathcal{H}$, the result follows by taking the limes
inferior.


\appendix
\section{Semiclassical States with a Slower von Neumann Entropy
  Renormalization}
\label[appendix]{sec:stat-with-diff}

In this appendix we construct an ``orthonormal system'' of coherent states
that both converges in topology $\mathcal{T}$ and in the sense of Wigner measures to
any Lebesgue-absolutely continuous measure $\mu\in \mathcal{P}(\mathcal{H})$, and for which the von
Neumann entropy converges to its classical counterpart, with a
renormalization that diverges \emph{more slowly} than the Gibbs
renormalization $d\log(\pi\varepsilon)$.


Firstly, given a set of distinct vectors $\{z_m\}_{m\in \mathbb{N}^{*}}\subset \mathcal{H}$, let us
construct an orthonormal system $\{e_{\varepsilon}(z_m)\}_{m\in \mathbb{N}^{*}}\subset \mathcal{F}_\varepsilon(\mathcal{H})$ out of the
coherent states $\{\lvert (z_m)_{\varepsilon}\rangle\}_{m\in \mathbb{N}^{*}}$; this is done by the recursive
Gram-Schmidt procedure:
\begin{itemize}
\item $e_{\varepsilon}(z_1)= \lvert (z_1)_{\varepsilon}\rangle$ ;
\item $e_{\varepsilon}(z_2)= \tfrac{1}{\lVert (z_2)_{\varepsilon} - \langle e_{\varepsilon}(z_1)\vert (z_2)_{\varepsilon}   \rangle_{\mathcal{H}}e_{\varepsilon}(z_1)
    \rVert_{\mathcal{H}}^{}}\Bigl(\lvert (z_2)_{\varepsilon}\rangle - \langle e_{\varepsilon}(z_1)\vert (z_2)_{\varepsilon}   \rangle_{\mathcal{H}}e_{\varepsilon}(z_1)\Bigr)$ ;
\item $\dotsm$
\item $e_{\varepsilon}(z_m)= \tfrac{1}{\lVert  (z_m)_{\varepsilon} - \sum_{j=1}^{m-1}\langle e_{\varepsilon}(z_j)\vert (z_m)_{\varepsilon}   \rangle_{\mathcal{H}}e_{\varepsilon}(z_j)
    \rVert_{\mathcal{H}}^{}}\Bigl(\lvert (z_m)_{\varepsilon}\rangle - \sum_{j=1}^{m-1}\langle e_{\varepsilon}(z_j)\vert (z_m)_{\varepsilon}   \rangle_{\mathcal{H}}e_{\varepsilon}(z_j)\Bigr)$ ;
\item $\dotsm$
\end{itemize}
Observing that
\begin{multline*}
  \Bigl\langle (z_2)_{\varepsilon}  - \langle z_1\vert z_2\rangle_{\mathcal{H}}(z_1)_{\varepsilon}  \:\Bigl\vert\:  W_{\varepsilon}(\zeta) \:\Bigr\vert\: (z_2)_{\varepsilon} - \langle z_1\vert z_2   \rangle_{\mathcal{H}}(z_1)_{\varepsilon}\Bigr\rangle_{\mathcal{F}_\varepsilon(\mathcal{H})}= \langle (z_2)_{\varepsilon}  \vert W_{\varepsilon}(\zeta) \vert (z_2)_{\varepsilon} \rangle_{\mathcal{F}_\varepsilon(\mathcal{H})} \\- e^{- \frac{\lVert z_1 - z_2  \rVert_{\mathcal{H}}^2}{2\varepsilon}}\Bigl(e^{\frac{i}{\varepsilon}\Im \langle z_1\vert z_2  \rangle_{\mathcal{H}}}\langle (z_2)_{\varepsilon}\vert W_{\varepsilon}(\zeta) \vert (z_1)_{\varepsilon}\rangle_{\mathcal{F}_\varepsilon(\mathcal{H})}+ e^{-\frac{i}{\varepsilon}\Im \langle z_1\vert z_2  \rangle_{\mathcal{H}}}\langle (z_1)_{\varepsilon}\vert W_{\varepsilon}(\zeta) \vert (z_2)_{\varepsilon}\rangle_{\mathcal{F}_\varepsilon(\mathcal{H})} \\- e^{- \frac{\lVert z_1 - z_2  \rVert_{\mathcal{H}}^2}{2\varepsilon}} \langle (z_1)_{\varepsilon}\vert W_{\varepsilon}(\zeta) \vert (z_1)_{\varepsilon}\rangle_{\mathcal{F}_\varepsilon(\mathcal{H})}\Bigr)\;,
\end{multline*}
and that
\begin{equation*}
  \lVert (z_2)_{\varepsilon} - \langle e_{\varepsilon}(z_1)\vert (z_2)_{\varepsilon}   \rangle_{\mathcal{H}}e_{\varepsilon}(z_1)\rVert_{\mathcal{H}}^{}=\sqrt{\bigl\langle (z_2)_{\varepsilon}  - \langle z_1\vert z_2\rangle_{\mathcal{H}}(z_1)_{\varepsilon}  \bigl\vert  W_{\varepsilon}(0) \bigr\vert (z_2)_{\varepsilon} - \langle z_1\vert z_2   \rangle_{\mathcal{H}}(z_1)_{\varepsilon}\bigr\rangle_{\mathcal{F}_\varepsilon(\mathcal{H})}}\;,
\end{equation*}
it is easy to prove the following lemma.

\begin{lemma}
  \label[lemma]{lemma:5}
  Given any set of distinct points $\{z_m\}_{m\in \mathbb{N}^{*}}\subset \mathcal{H}$, we have that for
  any $m\in \mathbb{N}^{*}$, and any $\zeta\in \mathcal{H}$,
  \begin{equation*}
    \langle e_{\varepsilon}(z_m) \vert W_{\varepsilon}(\zeta)\vert e_{\varepsilon}(z_m)  \rangle_{\mathcal{F}_\varepsilon(\mathcal{H})}= \tfrac{1}{N_{\varepsilon}^2} \Bigl(\langle (z_m)_{\varepsilon} \vert W_{\varepsilon}(\zeta)\vert (z_m)_{\varepsilon}  \rangle_{\mathcal{F}_\varepsilon(\mathcal{H})} + m\sum_{j=1}^{m-1} e^{-\frac{\lVert z_m-z_j  \rVert_{\mathcal{H}}^2}{2\varepsilon}}A_{\varepsilon}(\zeta)\Bigr)\;,
  \end{equation*}
  where $\lvert A_{\varepsilon}(\zeta)  \rvert_{}^{}\leq 1$ uniformly w.r.t.\ $\varepsilon$ and $\zeta$, and
  \begin{equation*}
    N_{\varepsilon}^2= 1+ m\sum_{j=1}^{m-1} e^{-\frac{\lVert z_m-z_j  \rVert_{\mathcal{H}}^2}{2\varepsilon}}A_{\varepsilon}(0)\;.
  \end{equation*}
\end{lemma}
\begin{corollary}
  \label[corollary]{cor:3}
  For any $m\in \mathbb{N}^{*}$,
  \begin{equation*}
    \lvert e(z_m)\rangle \langle e(z_m)\rvert \underset{\varepsilon\to 0}{\longrightarrow} \delta_{z_m}
  \end{equation*}
  in the sense of Wigner measures, furthermore, the Fourier transforms
  converge in topology $\mathcal{T}$.
\end{corollary}

Let us now define the family of dyadic lattices $\{\Lambda_M\}_{M\in \mathbb{N}}$, $\Lambda_M\subset \mathcal{H}$,
as follows: fix an orthonormal basis $\{\phi_j\}_{j=1}^d\subset \mathcal{H}$, and define $\zeta_j= \langle
\phi_j \vert z \rangle_{\mathcal{H}}$ the component of $z\in \mathcal{H}$ along $\phi_j$. Then
\begin{equation*}
  \Lambda_M= \Bigl\{z\in \mathcal{H}\;,\; (\zeta_1,\dotsc,\zeta_d)\in 2^{-M}(\mathbb{Z}\cap [-M2^M,M2^M])^d \Bigr\}\;.
\end{equation*}
In other words, $\Lambda_M$ is the set of vertices of dyadic cubes of side $2^{-M}$
inside the cube of side $2M$. Therefore, we have that $\Lambda_M\subset \Lambda_N$ whenever
$M<N$.

For convenience, let us enumerate each finite set $\Lambda_M$ consistently with
respect to the increasing family of sets $\{\Lambda_M\}_{M\in \mathbb{N}}$: choose any
enumeration for $\Lambda_0$; then recursively, given the enumeration of $\Lambda_M$,
\begin{equation*}
  \Lambda_M= \bigl\{z_m\bigr\}_{m=1}^{\lvert \Lambda_M  \rvert_{}^{}}\;,
\end{equation*}
choose any enumeration of $\Lambda_{M+1}\supset \Lambda_M$ that satisfies
\begin{equation*}
  \Lambda_{M+1}= \bigl\{z_m\bigr\}_{m=1}^{\lvert \Lambda_M  \rvert_{}^{}}\cup \bigl\{z_{m'}\bigr\}_{m'=\lvert \Lambda_M  \rvert_{}^{}+1}^{\lvert \Lambda_{M+1}  \rvert_{}^{}}\;.
\end{equation*}
For any measure $\mathcal{P}(\mathcal{H})\ni \mathrm{d}\mu(z)= f(z)\mathrm{d}z$ with $f$ regulated, we now
define the family of quantum states $\rho_{M,\varepsilon}(\mu)$ as
\begin{equation*}
  \rho_{M,\varepsilon}(\mu)=\sum_{z\in \Lambda_M}^{} \frac{f(z)}{2^{dM} N_M} \bigl\lvert e_{\varepsilon}(z)\bigr\rangle\bigl\langle e_{\varepsilon}(z)\bigr\rvert= \sum_{m=1}^{\lvert \Lambda_M  \rvert_{}^{}} \frac{f(z_m)}{2^{dM}N_M} \bigl\lvert e_{\varepsilon}(z_m)\bigr\rangle\bigl\langle e_{\varepsilon}(z_m)\bigr\rvert\;;
\end{equation*}
where $N_M$ is the normalization factor
\begin{equation*}
  N_M= \sum_{z\in \Lambda_M}^{} \frac{f(z)}{2^{dM}}\;.
\end{equation*}
Observe that $\rho_{M,\varepsilon}$ depends on $\varepsilon$ through the coherent orthonormal system
$e_{\varepsilon}(z)$. Using \cref{lemma:5}, we have that
\begin{multline*}
  \Tr_{\mathcal{F}_\varepsilon(\mathcal{H})}\bigl(\rho_{M,\varepsilon}(\mu) W_{\varepsilon}(\zeta)\bigr)= \sum_{m=1}^{\lvert \Lambda_M  \rvert_{}^{}}\frac{f(z_m)}{2^{dM}N_M} \langle e_{\varepsilon}(z_m)\vert W_{\varepsilon}(\zeta)\vert e_{\varepsilon}(z_m)  \rangle_{\mathcal{F}_\varepsilon(\mathcal{H})}\\= \sum_{m=1}^{\lvert \Lambda_M  \rvert_{}^{}}\tfrac{f(z_m)}{2^{dM}N_MN_{\varepsilon}^2} \Bigl(\langle (z_m)_{\varepsilon} \vert W_{\varepsilon}(\zeta)\vert (z_m)_{\varepsilon}  \rangle_{\mathcal{F}_\varepsilon(\mathcal{H})} + m\sum_{j=1}^{m-1} e^{-\frac{\lVert z_m-z_j  \rVert_{\mathcal{H}}^2}{2\varepsilon}}A_{\varepsilon}(\zeta)\Bigr)\;.
\end{multline*}
Furthermore, the noncommutative Fourier transform of a coherent state $\vert z_{\varepsilon}\rangle$ can be
computed explicitly, yielding
\begin{multline}
  \label{eq:2}
  \Tr_{\mathcal{F}_\varepsilon(\mathcal{H})}\bigl(\rho_{M,\varepsilon}(\mu) W_{\varepsilon}(\zeta)\bigr)\\=\sum_{m=1}^{\lvert \Lambda_M  \rvert_{}^{}}\tfrac{f(z_m)}{2^{dM}N_MN_{\varepsilon}^2} \Bigl(e^{-\frac{\varepsilon\lVert z_m  \rVert_{\mathcal{H}}^2}{2}}e^{2i \Re \langle \zeta\vert z_m  \rangle_{\mathcal{H}}} + m\sum_{j=1}^{m-1} e^{-\frac{\lVert z_m-z_j  \rVert_{\mathcal{H}}^2}{2\varepsilon}}A_{\varepsilon}(\zeta)\Bigr)\;.
\end{multline}
Focus now on the first term on the r.h.s.\ of \eqref{eq:2}:
\begin{multline}
  \label{eq:3}
  \sum_{m=1}^{\lvert \Lambda_M  \rvert_{}^{}}\tfrac{f(z_m)}{2^{dM}N_MN_{\varepsilon}^2} e^{-\frac{\varepsilon\lVert z_m  \rVert_{\mathcal{H}}^2}{2}}e^{2i \Re \langle \zeta\vert z_m  \rangle_{\mathcal{H}}}\\= \tfrac{1}{N_M}\sum_{m=1}^{\lvert \Lambda_M  \rvert_{}^{}}\tfrac{f(z_m)}{2^{dM}N_{\varepsilon}^2} e^{2i \Re \langle \zeta\vert z_m  \rangle_{\mathcal{H}}} + \sum_{m=1}^{\lvert \Lambda_M  \rvert_{}^{}}\tfrac{f(z_m)}{2^{dM}N_MN_{\varepsilon}^2} e^{2i \Re \langle \zeta\vert z_m  \rangle_{\mathcal{H}}}\Bigl(e^{-\frac{\varepsilon\lVert z_m  \rVert_{\mathcal{H}}^2}{2}}-1\Bigr)\;.
\end{multline}
The first term on the r.h.s.\ of \eqref{eq:3} can be in turn written as
\begin{equation*}
  \tfrac{1}{N_M}\sum_{m=1}^{\lvert \Lambda_M  \rvert_{}^{}}\tfrac{f(z_m)}{2^{dM}N_{\varepsilon}^2} e^{2i \Re \langle \zeta\vert z_m  \rangle_{\mathcal{H}}} = \tfrac{1}{N_M}\sum_{m=1}^{\lvert \Lambda_M  \rvert_{}^{}}\tfrac{f(z_m)}{2^{dM}} e^{2i \Re \langle \zeta\vert z_m  \rangle_{\mathcal{H}}} +\tfrac{1}{N_M}\sum_{m=1}^{\lvert \Lambda_M  \rvert_{}^{}}\tfrac{f(z_m)}{2^{dM}} e^{2i \Re \langle \zeta\vert z_m  \rangle_{\mathcal{H}}}\Bigl(\tfrac{1}{N_{\varepsilon}^2}-1\Bigr)\;.
\end{equation*}
Now, by construction we have that
\begin{equation*}
  \lim_{M\to \infty}\sum_{m=1}^{\lvert \Lambda_M  \rvert_{}^{}}\tfrac{f(z_m)}{2^{dM}} e^{2i \Re \langle \zeta\vert z_m  \rangle_{\mathcal{H}}}= \int_{\mathcal{H}}^{} e^{2i\Re\langle \zeta\vert z  \rangle_{\mathcal{H}}}  f(z)\mathrm{d}z
\end{equation*}
and
\begin{equation*}
  \lim_{M\to \infty} N_M= \int_{\mathcal{H}}^{}f(z)  \mathrm{d}z=1\;;
\end{equation*}
while by \cref{lemma:5}
\begin{equation*}
  \lvert N_{\varepsilon}^2-1  \rvert_{}^{}\leq  m\sum_{j=1}^{m-1} e^{- \frac{\lVert z_m-z_j  \rVert_{\mathcal{H}}^2}{2\varepsilon}}\leq m^2 e^{- \frac{1}{2^{2M+1}\varepsilon}}\leq 2^{2d}e^{d\log M + dM\log 2 -\varepsilon^{-1}2^{-2M-1}}\;. 
\end{equation*}
Therefore, for any $\varepsilon(M)\ll 2^{-2M}/M$, we obtain
\begin{equation*}
  \lim_{M\to \infty}\tfrac{1}{N_M}\sum_{m=1}^{\lvert \Lambda_M  \rvert_{}^{}}\tfrac{f(z_m)}{2^{dM}N_{\varepsilon(M)}^2} e^{2i \Re \langle \zeta\vert z_m  \rangle_{\mathcal{H}}}= \int_{\mathcal{H}}^{} e^{2i\Re\langle \zeta\vert z  \rangle_{\mathcal{H}}}  f(z)\mathrm{d}z\;.
\end{equation*}
The second term on the r.h.s.\ of \eqref{eq:3} can be similarly bounded:
\begin{equation*}
  \sum_{m=1}^{\lvert \Lambda_M  \rvert_{}^{}}\frac{f(z_m)}{2^{dM}N_M N_{\varepsilon}^2}\Bigl\lvert e^{-\frac{\varepsilon\lVert z_m  \rVert_{\mathcal{H}}^2}{2}}-1  \Bigr\rvert\leq \Bigl\lvert e^{-\frac{\varepsilon M^2}{2}}-1  \Bigr\rvert \sum_{m=1}^{\lvert \Lambda_M  \rvert_{}^{}}\frac{f(z_m)}{2^{dM}N_M N_{\varepsilon}^2}\;,
\end{equation*}
and therefore
\begin{equation*}
  \lim_{M\to \infty}\sum_{m=1}^{\lvert \Lambda_M  \rvert_{}^{}}\tfrac{f(z_m)}{2^{dM}N_MN_{\varepsilon(M)}^2} e^{2i \Re \langle \zeta\vert z_m  \rangle_{\mathcal{H}}}\Bigl(e^{-\frac{\varepsilon(M)\lVert z_m  \rVert_{\mathcal{H}}^2}{2}}-1\Bigr)=0
\end{equation*}
whenever $\varepsilon(M)M^2\to 0$ (that is implied by $\varepsilon(M)\ll 2^{-2M}/M$).

Consider finally the second term on the r.h.s.\ of \eqref{eq:2}:
\begin{equation*}
  \sum_{m=1}^{\lvert \Lambda_M  \rvert_{}^{}}m\sum_{j=1}^{m-1} e^{-\frac{\lVert z_m-z_j  \rVert_{\mathcal{H}}^2}{2\varepsilon}}A_{\varepsilon}(\zeta)\;,
\end{equation*}
again an analogous argument shows that
\begin{equation*}
  \lim_{M\to \infty}\sum_{m=1}^{\lvert \Lambda_M  \rvert_{}^{}}m\sum_{j=1}^{m-1} e^{-\frac{\lVert z_m-z_j  \rVert_{\mathcal{H}}^2}{2\varepsilon(M)}}A_{\varepsilon(M)}(\zeta)=0
\end{equation*}
whenever $d\log M + dM\log 2 -\varepsilon(M)^{-1}2^{-2M-1}\to -\infty$. Summing up the above
discussion, one proves the following lemma.
\begin{lemma}
  \label[lemma]{lemma:6}
  The sequence of quantum states $\rho_{M,\varepsilon(M)}(\mu)$ converges to the absolutely
  continuous measure $\mu\in \mathcal{P}(\mathcal{H})$ in the sense of Wigner measures as $M\to \infty$, for
  any choice $\varepsilon(M)\ll 2^{-2M}/M$.
\end{lemma}
\begin{corollary}
  \label[corollary]{cor:4}
  Since $\Tr_{\mathcal{F}_\varepsilon(\mathcal{H})}\bigl(\rho_{M,\varepsilon(M)}(\mu)N_{\varepsilon(M)}\bigr)\leq C$ uniformly with respect
  to $M$, it follows that
  \begin{equation*}
    \widehat{\rho}_{M,\varepsilon(M)}(\mu)\xrightarrow[M\to \infty]{\mathcal{T}}\widehat{\mu}\;;
  \end{equation*}
  and
  \begin{equation*}
    \lim_{M\to \infty}\Tr_{\mathcal{F}_\varepsilon(\mathcal{H})}\bigl(\rho_{M,\varepsilon(M)}(\mu)H_{\varepsilon(M)}\bigr) = \int_{\mathcal{H}}^{}h(z) f(z)  \mathrm{d}z\;.
  \end{equation*}
\end{corollary}

The von Neumann entropy of $\rho_{M,\varepsilon}$ is straightforward to compute:
\begin{equation*}
  S_{\mathrm{vN},\varepsilon}\bigl(\rho_{M,\varepsilon}(\mu)\bigr)= -\tfrac{1}{N_M}\sum_{m=1}^{\lvert \Lambda_M  \rvert_{}^{}} \tfrac{f(z_m)}{2^{dM}}\log f(z_m) +\log N_M+ dM\log 2\;,
\end{equation*}
from which we get that
\begin{equation*}
  \lim_{M\to \infty}S_{\mathrm{vN},\varepsilon}\bigl(\rho_{M,\varepsilon}(\mu)\bigr) - dM\log 2 = S_{\mathrm{B}}(\mu)\;. 
\end{equation*}
However, the relative von Neumann entropy
$S_{\mathrm{vN},\varepsilon}\bigl(\rho_{M,\varepsilon}(\mu)\Vert\Gamma_{\beta,\varepsilon}\bigr)$ reads
\begin{multline*}
  S_{\mathrm{vN},\varepsilon}\bigl(\rho_{M,\varepsilon}(\mu)\Vert\Gamma_{\beta,\varepsilon}\bigr)= \tfrac{1}{N_M}\sum_{m=1}^{\lvert \Lambda_M  \rvert_{}^{}} \tfrac{f(z_m)}{2^{dM}}\log f(z_m) \\-\log N_M +\beta \Tr_{\mathcal{F}_\varepsilon(\mathcal{H})}\bigl(\rho_{M,\varepsilon}(\mu)H_{\varepsilon}\bigr)+ d\bigl(-\log (\pi\varepsilon)-M\log 2\bigr)\;.
\end{multline*}
It is then clear that for any $\varepsilon(M)\ll 2^{-2M}/M$, the relative entropy
diverges, for the renormalization factor $-\log(\pi\varepsilon(M))$ has the ``wrong
constant'' in front of the logarithmic divergence: in particular, for $\varepsilon(M) =
\frac{2^{-(2+\delta)M}}{\pi}$ and $M\gg1$, we have
\begin{equation*}
  S_{\mathrm{vN},\varepsilon(M)}\bigl(\rho_{M,\varepsilon(M)}(\mu)\Vert\Gamma_{\beta,\varepsilon(M)}\bigr)\sim S_{\mathrm{B}}\bigl(x\Vert\gamma_{\beta}\bigr) + d(1+\delta)M \log 2 \;.
\end{equation*}

\end{document}